\documentclass[10pt,a4paper]{article}
\usepackage[utf8]{inputenc}
\usepackage{amsmath, amsthm, amssymb}
\usepackage{enumerate}
\usepackage{cite}
\newtheorem{thm}{Theorem}[section]

\newtheorem{tvrzx}[thm]{Proposition}
\newenvironment{tvrz}{\begin{tvrzx}}{\smallskip\end{tvrzx}}

\newtheorem{lemmax}[thm]{Lemma}
\newenvironment{lemma}{\begin{lemmax}}{\smallskip\end{lemmax}}

\newtheorem{theoremx}[thm]{Theorem}

\theoremstyle{definition}
\newtheorem{definicex}[thm]{Definition}
\newenvironment{definice}{\begin{definicex}}{\medskip\end{definicex}}

\theoremstyle{remark}
\newtheorem{remx}[thm]{Remark}
\newenvironment{rem}{\begin{remx}}{\medskip\end{remx}}

\theoremstyle{definition}
\newtheorem{examplex}[thm]{Example}

\def\R{\mathbb{R}}

\def\T{\mathbb{T}}
\def\<{\langle}
\def\>{\rangle}
\def\~{\widetilde}
\def\^{\wedge}

\def\g{\mathfrak{g}}

\def\io{\mathit{i}}
\def\G{\mathcal{G}}
\def\H{\mathcal{H}}
\def\B{\mathcal{B}}
\def\P{\mathcal{P}}

\def\T{\mathcal{T}}
\def\nc{\bullet}
\def\c{\circ}
\def\Fi{{\text{\bfseries{\scriptsize{I}}}}}
\def\Fii{\text{\bfseries{\scriptsize{II}}}}
\def\NB{N}
\newcommand{\bi}[1]{ {\mathbf{ #1 }} }
\newcommand{\TM}[1]{\Lambda^{#1}TM}
\newcommand{\cTM}[1]{\Lambda^{#1}T^{\ast}M}
\newcommand{\vf}[1]{\mathfrak{X}^{#1}(M)}
\newcommand{\df}[1]{\Omega^{#1}(M)}
\newcommand{\ppx}[2]{\frac{\partial{#1}}{\partial{x^{#2}}}}
\newcommand{\ppy}[2]{\frac{\partial{#1}}{\partial{y^{#2}}}}

\newcommand{\bm}[4]{\begin{pmatrix}#1 & #2 \\ #3 & #4 \end{pmatrix}}
\newcommand{\bv}[2]{\begin{pmatrix}#1\\#2\end{pmatrix}}

\newcommand{\Fud}[2]{ {\widehat{F}{}^{#1}}_{#2} }
\newcommand{\Fdu}[2]{ {\widehat{F}_{#1}}{}^{#2} }
\newcommand{\cv}[2]{ \begin{pmatrix}#1\\#2\end{pmatrix} }

\newcommand{\px}[1]{{\~{\partial X}}^{#1}}
\DeclareMathOperator{\diag}{diag}
\DeclareMathOperator{\tr}{tr}
\DeclareMathOperator{\Img}{Im}
\DeclareMathOperator{\Diff}{Diff}
\begin{document}
\begin{flushright}
\today
\end{flushright}
\vspace{0.7cm}
\begin{center}

\baselineskip=13pt {\Large \bf{Extended generalized geometry and a DBI-type effective action for branes ending on branes}\\}
 \vskip1.3cm
 Branislav Jur\v co$^1$, Peter Schupp$^2$, Jan Vysoký$^{2,3}$\\
 \vskip0.6cm
$^{1}$\textit{Charles University in Prague, Faculty of Mathematics and Physics, Mathematical Institute 
\\ Prague 186 75, Czech Republic, jurco@karlin.mff.cuni.cz}\\
\vskip0.3cm
$^{2}$\textit{Jacobs University Bremen\\ 28759 Bremen, Germany, p.schupp@jacobs-university.de}\\
\vskip0.3cm
$^{3}$\textit{Czech Technical University in Prague\\ Faculty of Nuclear Sciences
and Physical Engineering\\ Prague 115 19, Czech Republic, vysokjan@fjfi.cvut.cz}\\
\vskip0.5cm
\end{center}
\vspace{0.4cm}

\begin{abstract}
Starting from the Nambu-Goto bosonic membrane action, we develop a geometric description suitable for $p$-brane backgrounds. With tools of generalized geometry we derive the pertinent generalization of the string open-closed relations  to the $p$-brane case. Nambu-Poisson structures are used in this context to generalize the concept of semiclassical noncommutativity of $D$-branes governed by a Poisson tensor. We find a natural description of the correspondence of recently proposed commutative and noncommutative versions of an effective action for $p$-branes ending on a $p'$-brane. We calculate the power series expansion of the action in background independent gauge. Leading terms in the double scaling limit are given by a generalization of a (semi-classical) matrix model.
\end{abstract}


\begin{flushright}
\textit{Dedicated to the memory of Julius Wess and Bruno Zumino}
\end{flushright}
\section{Introduction}

Among the most intriguing features of fundamental theories of extended objects are novel types of symmetries  and concomitant generalized notions of geometry. Particularly interesting examples of these symmetries are T-duality in closed string theory and the equivalence of commutative/noncommutative descriptions in open string theory. These symmetries have their natural settings in generalized geometry and noncommutative geometry. Low energy effective theories link the fundamental theories to potentially observable phenomena in (target) spacetime. Interestingly, the spacetime remnants of the stringy symmetries can fix these effective theories essentially uniquely without the need of actual string computations: ``string theory with no strings attached.''

The main objective of this paper is to study this interplay of symmetry and geometry in the case of higher dimensional extended objects (branes). 
More precisely, we intended to extend, clarify and further develop the construction outlined in \cite{Jurco:2012yv} that tackles the quest to find an all-order effective action for a system of multiple $p$-branes ending on a $p'$-brane.
The result for the case of open strings ending on a single D-brane is well known: The Dirac-Born-Infeld action provides an effective description to all orders in $\alpha'$ \cite{Fradkin:1985ys,Leigh:1989jq,Tseytlin:1999dj}. The way that this effective action has originally been derived from first principles in string theory is rather indirect: The effective action is determined by requiring that its equations of motion double as consistency conditions for an anomaly free world sheet quantization of the fundamental string. A more direct target space approach can be based on T-duality arguments.
Moreover, there is are equivalent commutative and non-commutative descriptions \cite{Seiberg:1999vs}, where the equivalency condition fixes the action essentially uniquely \cite{Cornalba:2000ua,Jurco:2001my}. This ``commutative-noncommutative duality" has been used also to study the non-abelian DBI action \cite{Terashima:2000ej,Cornalba:2000ua}. In the context of the M2/M5 brane system a generalization has been proposed in\cite{Chen:2010br}. 

In this paper, we focus only on the bosonic part of the action. The main idea of \cite{Jurco:2012yv}, inspired by \cite{Chen:2010br}, was to introduce open-closed membrane relations, and a Nambu-Poisson map which can be used to relate ordinary higher gauge theory to a new Nambu gauge theory \cite{Ho:2013iia,Jurco:2014aza,Ho:2013opa,Ho:2013paa}. See also the work of P.-M. Ho et al. \cite{Ho:2007vk,Ho:2008nn,Ho:2008ve,Ho:2009zt} and K. Furuuchi et al. \cite{Furuuchi:2009zx,Furuuchi:2010sp} on relation of M2/M5 to Nambu-Poisson structures.
It turns out that the requirement of ``commutative-noncommutative duality" determines the bosonic part of the effective action essentially uniquely. Interesting open problems are to determine, in the case of a M5-brane, the form of the full supersymmetric action and to check consistency with $\kappa$-symmetry and (nonlinear) selfduality.

Nambu-Poisson structures were first considered by Y. Nambu already in 1973 \cite{1973PhRvD...7.2405N}, and generalized and axiomatized more then 20 years later by L. Takhtajan \cite{Takhtajan:1993vr}. The axioms of Nambu-Poisson structures, although they seem to be a direct generalization of Poisson structures, are in fact very restrictive. This was already conjectured in the pioneering paper \cite{Takhtajan:1993vr} and proved three years later in \cite{decomposability,Gautheron}. For a modern treatment of Nambu-Poisson structures see \cite{hagiwara,2011ScChA..54..437B,2010arXiv1003.1004Z}.

Matrix-model like actions using Nambu-Poisson structures are a current focus of research (see e.g. \cite{Park:2008qe,Sato:2010ca,DeBellis:2010sy, Chu:2011yd}) motivated by the works of \cite{Basu:2004ed,SheikhJabbari:2005mf,PhysRevD.75.045020,Bagger:2007jr,Gustavsson:2007vu} and others. See also \cite{Maldacena:2002rb, SheikhJabbari:2004ik} for further reference. 
Among the early approaches, the one closest to ours is the one of \cite{Cederwall:1997gg,Bao:2006ef}, which uses $\kappa$-symmetry as a guiding principle and features a non-linear self-duality condition. It avoids the use of an auxiliary chiral scalar \cite{Pasti:1997gx} with its covariance problems following a suggestion of \cite{Witten:1996hc}. For these and alternative formulations, e.g., those of \cite{Howe:1996yn}, based on superspace embedding and $\kappa$-symmetry, we refer to the reviews \cite{Sorokin:1999jx,Simon:2011rw}.

Generalized geometry was introduced by N. Hitchin in \cite{Hitchin:2004ut,2005math......8618H,2006CMaPh.265..131H}. It was further elaborated in \cite{Gualtieri:2003dx}. Although Hitchin certainly recognized the possible importance for string backgrounds, and commented on it in \cite{Hitchin:2004ut}, this direction is not pursued there.
Recently, a focus of applications of generalized geometry, is superstring theory and supergravity. Here we mention closely related work \cite{Coimbra:2011nw,Coimbra:2012af}. The role of generalized geometries in M-theory was previously examined by C.M. Hull in \cite{Hull:2007zu}.
A further focus is the construction of the field theories based on objects of generalized geometry. This is mainly pursued in \cite{Kotov:2004wz,2005JGP....54..400B} and in \cite{Kotov:2010wr}, see also \cite{Zucchini:2005rh}. Generalized geometry (mostly Courant algebroid brackets) was also used in relation to worldsheet algebras and non-geometric backgrounds. See, for example, \cite{Alekseev:2004np,Bonelli:2005ti,2011JHEP...03..074E} and \cite{Halmagyi:2009te,Halmagyi:2008dr}. One should also mention the use of generalized geometry in the description of T-duality, see\cite{2011arXiv1106.1747C}, or the lecture notes \cite{Bouwknegt:2010zz}. An outline of the relation of T-duality with generalized geometry can be found in \cite{Grana:2008yw}. Finally, there is an interesting interpretation of D-branes in string theory as Dirac structures of generalized geometry in \cite{Asakawa:2012px,Asakawa:2014eva}.
Finally, in \cite{Jurco:2013upa}, we have used generalized geometry to describe the relation between string theory and non-commutative geometry.

This paper is organized as follows:

In section \ref{sec_membrane}, we review  basic facts concerning classical membrane actions. In particular, we recall how gauge fixing  can be used to find a convenient form of the action.  We show that the corresponding Hamiltonian density is a fiberwise metric on a certain vector bundle. We present background field redefinitions, generalizing the well-known open-closed relations of Seiberg and Witten.

In section \ref{sec_nambu}, we describe the sigma model dual to the membrane action. It is a straightforward generalization of the non-topological Poisson sigma model of the $p=1$ case.

Section \ref{sec_geometry} sets up the geometrical framework for the field redefinitions of the previous sections. An extension  of  generalized geometry is used to describe open-closed relations as an orthogonal transformation of the generalized metric on the vector bundle $TM \oplus \TM{p} \oplus T^{\ast}M \oplus \cTM{p}$. Compared to the $p=1$ string case, we find the need for a second ``doubling'' of the geometry. The split in $TM$ and $\TM{p}$ has its origin in  gauge fixing of the auxiliary metric on the \mbox{$p+1$}-dimensional brane world volume and the two parts are related to the temporal and spatial worldvolume directions. To the best of our knowledge, this particular structure $W \oplus W^*$ with $W = TM \oplus \TM{p}$ has not been considered in the context of M-theory before.

In section \ref{sec_gaugeF}, we introduce the $(p+1)$-form gauge field $F$ as a fluctuation of the original membrane background. We show that this can be viewed as an orthogonal transformation of the generalized metric describing the membrane backgrounds. On the other hand, the original background can equivalently be described in terms of open variables and this description can be extended to include fluctuations. Algebraic manipulations are used to identify the pertinent background fields. The construction requires the introduction of a target manifold diffeomorphism, which generalizes the (semi-classical) Seiberg-Witten map from the string to the $p>1$ brane case.

This map is explicitly constructed in section \ref{sec_SWmap} using a generalization of Moser's lemma. The key ingredient is the fact that $\Pi$, which appears in the open-closed relations, can be chosen to be a Nambu-Poisson tensor. 
Attention is paid to a correct mathematical formulation of the analogue of a symplectic volume form for Nambu-Poisson structures.

Based on the results of the previous sections, we prove in  section \ref{DBI} the equivalence of a commutative and semiclassically noncommutative DBI action. We present various forms of the same action using determinant identities of block matrices.  Finally, we compare our action to existing proposals for the M5-brane action.

In section \ref{sec_BIG}, we show that the Nambu-Poisson structure $\Pi$ can be chosen to be the pseudoinverse of the $(p+1)$-form background field $C$. In analogy with the $p=1$ case, we call this choice ``background independent gauge''. However, for $p>1$ we have to consider both algebraic and geometric properties of $C$ in order to obtain a well defined Nambu-Poisson tensor $\Pi$. The generalized geometry formalism developed in section \ref{sec_geometry} is used to derive the results in a way that looks formally  identical to the much easier $p=1$ case. (This is a nice example of the power of generalized geometry.)

In section \ref{sec_NCdirections}, we introduce a convenient splitting of the tangent bundle and rewrite all membrane backgrounds in coordinates adapted to this splitting using a block matrix formalism. We introduce an appropriate generalization of the double scaling limit of \cite{Seiberg:1999vs} to cut off the series expansion of the effective action.

In the final section \ref{sec_matrix} of the paper, we use background independent gauge, double scaling limit, and coordinates adapted to the non-commutative directions to expand the DBI action up to first order in the scaling parameter. It turns out that this double scaling limit cuts off the infinite series in a physically meaningful way. We identify a possible candidate for the generalization of a matrix model.
For a discussion of the underlying Nambu-Poisson gauge theory we refer to \cite{Jurco:2014aza}.

\section{Conventions}

Thorough the paper, $p>0$ is a fixed positive integer. Furthermore, we assume that we are given a  $(p+1)$-dimensional compact orientable
worldvolume $\Sigma$ with local coordinates
$(\sigma^{0}, \dots, \sigma^{p})$. We may interpret $\sigma^{0}$ as a time parameter. Integration over all coordinates is
indicated by $\int d^{p+1}\sigma$, whereas the integration over
space coordinates $(\sigma^{1}, \dots, \sigma^{p})$ is indicted
as $\int d^{p}\sigma$. Indices corresponding to the worldvolume
coordinates are denoted by Greek characters $\alpha, \beta, \dots$, etc.
As usual, $\partial_{\alpha} \equiv \frac{\partial}{\partial
\sigma^{\alpha}}$.
We assume that the $n$-dimensional target manifold $M$ is equipped
with a set of local coordinates $(y^{1},\dots,y^{n})$. We denote the
corresponding indices by lower case Latin characters $i,j,k, \dots$, etc.
Upper case Latin characters $I,J,K, \dots$, etc. will denote strictly
ordered $p$-tuples of indices corresponding to $(y)$ coordinates,
e.g., $I = (i_{1}, \dots, i_{p})$ with $1\leq i_{1} < \dots <
i_{p}\leq n$. We use the shorthand notation
$\partial_{J} \equiv \ppy{}{j_{1}} \^ \dots \^ \ppy{}{j_{p}}$ and
$dy^{J} = dy^{j_{1}} \^ \dots \^dy^{j_{p}}$. The degree $q$-parts of the exterior
algebras of vector fields $\vf{}$ and forms  $\Omega(M)$ are denoted by $\vf{q}$ and $\df{q}$, respectively.

Where-ever a metric $g$ on $M$ is introduced, we assume that it is positive definite, i.e.,  $(M,g)$ is a Riemannian manifold. With this choice we will find a natural interpretation of membrane backgrounds in terms of generalized geometry. For any metric tensor $g_{ij}$, we denote, as usually, by $g^{ij}$
the components of the inverse contravariant tensor.

We use the following convention to handle $(p+1)$-tensors on $M$.
Let $B \in \df{p+1}$ be a $(p+1)$-form on $M$. We define the
corresponding vector bundle map $B_{\flat}: \TM{p} \rightarrow
T^{\ast}M$ as $B_{\flat}(Q) = B_{iJ} Q^{J} dy^{i}$, where $Q = Q^{J}
\partial_{J}$. We do not distinguish between vector bundle morphisms
and the induced $C^{\infty}(M)$-linear maps of smooth sections. We
will usually use the letter $B$ also for the $\binom{n}{p} \times n$
matrix of $B_{\flat}$ in the local basis $\partial_{J}$ of $\vf{p}$
and $dy^{i}$ of $\df{1}$, that is $(B)_{i,J} = \< \partial_{i},
B_{\flat}(\partial_{J}) \>$. Similarly, let $\Pi \in \vf{p+1}$; the
induced map $\Pi^{\sharp}: \cTM{p} \rightarrow TM$ is defined as
$\Pi^{\sharp}(\xi) = \Pi^{iJ} \xi_{J} \partial_{i}$ for $\xi =
\xi_{J} dy^{J}$. We use the letter $\Pi$ also for the $\binom{n}{p}
\times n$ matrix of $\Pi^{\sharp}$, that is $(\Pi)^{i,J} = \<dy^{i},
\Pi^{\sharp}(dy^{J})\>$. Clearly, with these conventions $(B)_{i,J} =
B_{iJ}$ and $(\Pi)^{i,J} = \Pi^{iJ}$.

Let $X: \Sigma \rightarrow M$ be a smooth map. We use the notation
$X^{i} = y^{i} \circ X$, and correspondingly $dX^{i} = d(X^{i}) =
X^{\ast}(dy^{i})$. Similarly, $dX^{J} = X^{\ast}(dy^{J})$. We reserve
the symbol $\px{J}$ for spatial components of the $p$-form $dX^{J}$,
that is, $\px{J} = (dX^{J})_{1 \dots p}$. We define the generalized
Kronecker delta $\delta_{i_{1} \dots i_{p}}^{j_{1}
\dots j_{p}}$ to be  $+1$ whenever the top $p$-index constitutes an
even permutation of the bottom one, $-1$ if for the odd permutation,
and $0$ otherwise. In other words, $\delta_{i_{1} \dots i_{p}}^{j_{1} \dots j_{p}} = p! \cdot \delta_{[i_{1}}^{[j_{1}} \dots \delta_{i_{p}]}^{j_{p}]}$. We use the convention $\epsilon_{i_{1} \dots
i_{p}} \equiv \epsilon^{i_{1} \dots i_{p}} \equiv \delta_{i_{1}
\dots i_{p}}^{1 \dots p} \equiv \delta_{1 \dots p}^{i_{1} \dots
i_{p}}$. Thus, in this notation we  have $\px{I} =
\partial_{l_{1}}X^{i_{1}} \cdots
\partial_{l_{p}}X^{i_{p}} \epsilon^{l_{1} \dots l_{p}}$.

\section{Membrane actions} \label{sec_membrane}

The most straightforward generalization of the relativistic string
action to higher dimensional world volumes is the Nambu-Goto $p$-brane action, simply measuring the
volume of the $p$-brane:
\begin{equation} \label{def_ngaction}
S_{NG}[X] = T_{p} \int d^{p+1}\sigma \sqrt {\det{( \partial_{\alpha} X^{i} \partial_{\beta} X^{j} g_{ij} )}},
\end{equation}
where $g_{ij}$ are components of the positive definite target  space metric $g$, and
$X: \Sigma \rightarrow M$ is the $n$-tuple of scalar fields
describing the $p$-brane. In a similar manner as for the string action,
one can introduce an auxiliary Riemannian metric $h$ on $\Sigma$ and
find the classically equivalent Polyakov action of the $p$-brane:

\begin{equation} \label{def_polyakov}
 S_{P}[X,h] = \frac{T'_{p}}{2} \int d^{p+1}\sigma \sqrt{h} \Big( h^{\alpha \beta} \partial_{\alpha}X^{i} \partial_{\beta}X^{j} g_{ij} - (p-1)\lambda \Big),
 \end{equation}
where $\lambda > 0$ can be chosen arbitrarily (but fixed), and $T'_{p}
= \lambda^{\frac{p-1}{2}} T_{p}$. Using the equations of motion for
$h^{\alpha \beta}$'s:
\begin{equation} \label{eq_eqmhab}
\frac{1}{2} h_{\alpha \beta} \big( h^{\gamma \delta} g_{\gamma \delta} - (p-1)\lambda \big) = g_{\alpha \beta},
\end{equation}
where $g_{\alpha \beta} = [X^{\ast}(g)]_{\alpha \beta} \equiv
\partial_{\alpha}X^{i} \partial_{\beta}X^{j} g_{ij}$, in $S_{P}$,
one gets back to (\ref{def_ngaction}).  In the rest of the paper, we will
choose $T_{p} \equiv 1$. Using reparametrization invariance, one
can always (at least locally) choose coordinates $(\sigma^{0},
\dots, \sigma^{p})$ such that $h_{00} = \lambda^{p-1} \det
h_{ab}$, $h_{0a} = 0$, where $h_{ab}$ denotes the space-like
components of the metric. 
In this gauge, the first term in
action (\ref{def_polyakov}) splits into two parts, one of them
containing only the spatial derivatives of $X^{i}$ and the spatial
components of the metric $h$. Using now the equations of motion for
$h_{ab}$, one gets the gauge fixed Polyakov action\footnote{The gauge constraints on $h_{a0}$, $h_{0b}$ and $h_{00}$ imply an energy-momentum tensor with vanishing components $T_{a0} = T_{0a}$ and $T_{00}$. These constraints must be considered along with the equations of motion of the action (\ref{def_polyakovgf}), to ensure equivalence with the actions (\ref{def_ngaction}) and (\ref{def_polyakov}). As discussed in \cite{Bars}, the subgroup of the diffeomorphism symmetries that remains after gauge fixing is a symmetry of the gauge-fixed p-brane action (\ref{def_polyakovgf}) and also transforms the pertinent components of the energy-momentum tensor into one another (even if they are not set equal to zero). The constraints can thus be consistently imposed at the level of
states.}
\begin{equation} \label{def_polyakovgf}
S_{P}^{gf}[X] = \frac{1}{2} \int d^{p+1}\sigma
\big\{\partial_{0}X^{i} \partial_{0}X^{j} g_{ij} +
\det{(\partial_{a}X^{i} \partial_{b}X^{j} g_{ij})} \big\}.
\end{equation}
The second term can be rewritten in a more convenient form once we
define
\begin{equation} \label{def_tensorofg}
\~g_{IJ} = \sum_{\pi \in \Sigma_{p}} sgn(\pi) g_{i_{\pi(1)} j_{1}}
\dots g_{i_{\pi(p)} j_{p}} \equiv  \delta^{k_{1} \dots k_{p}}_{I}
g_{k_1 j_1} \dots g_{k_p j_p}.
\end{equation}
Using this notation, one can write
\begin{equation} \label{def_polyakovgf2}
S_{P}^{gf}[X] = \frac{1}{2} \int d^{p+1}\sigma \big\{ \partial_{0}X^{i} \partial_{0}X^{j} g_{ij} +
\px{I} \px{J} \~g_{IJ} \big\}.
\end{equation}
From now on, assume that $g$ is a positive definite metric on $M$.
Note that from the symmetry of $g$ it follows that $\~g_{IJ} =
\~g_{JI}$.  We can view $\~g$ as a fibrewise bilinear form on the vector
bundle $\TM{p}$. Moreover, at any $m \in M$, one can define the
basis $(E_{I})$ of $\Lambda^{p} T_{m}M$ as $E_{I} = e_{i_{1}} \^
\dots \^ e_{i_{p}}$, where $(e_{1}, \dots, e_{n})$ is the
orthonormal basis for the quadratic form $g(m)$ at $m \in M$. In this
basis one has $\~g(m)(E_{I},E_{J}) = \delta_{I,J}$, which shows that $\~g$ is a positive definite fibrewise metric on $\TM{p}$.

For any $C \in \df{p+1}$, we can add the following coupling term to the action:
\begin{equation}
S_{C}[X] = - i \int_{\Sigma} X^{\ast}(C) = -i \int d^{p+1}\sigma \partial_{0}X^{i} \px{J} C_{iJ}.
\end{equation}
The resulting gauge fixed Polyakov action $S_{P}^{tot}[X] = S_{P}^{gf}[X] + S_{C}[X]$ has the form
\begin{equation} \label{def_actiontot}
S_{P}^{tot}[X] = \frac{1}{2} \int d^{p+1}\sigma \big\{ \partial_{0}X^{i}
\partial_{0}X^{j} g_{ij} +
\px{I} \px{J} \~g_{IJ} - 2i \partial_{0}X^{i} \px{J} C_{iJ} \big\}.
\end{equation}
This can be written in the compact matrix form by defining an ($n +
\binom{n}{p}$)-row vector 
\[
\Psi = \bv{i\partial_{0}X^{i}}{\px{J}} .
\]
The action then has the block matrix form
\begin{equation} \label{def_actiontotmatrix}
S_{P}^{tot}[X] = \frac{1}{2} \int d^{p+1}\sigma \{ \Psi^{\dagger} \bm{g}{C}{-C^{T}}{\~g} \Psi \}.
\end{equation}

From now on, unless explicitly mentioned, we may assume that $\~g$
is not necessarily of the form (\ref{def_tensorofg}), i.e., $\~g$
can be any positive definite fibrewise metric on $\TM{p}$. Any
further discussions will, of course, be valid also for the special
case (\ref{def_tensorofg}). Since $g$ is non-degenerate, we can pass
from the Lagrangian to the Hamiltonian formalism and vice versa. The
corresponding Hamiltonian has the form
\begin{equation} \label{def_polyakovham}
H_{P}^{tot}[X,P] = -\frac{1}{2} \int d^{p}\sigma \bv{iP}{\px{}}^{T}
\bm{g^{-1}}{-g^{-1}C}{-C^{T}g^{-1}}{\~g + C^{T}g^{-1}C}
\bv{iP}{\px{}}.
\end{equation}
The expression  $\~g + C^{T}g^{-1}C$ in the Hamiltonian and a similar expression $g + C\~g^{-1}C^{T}$ play the role of ``open membrane metrics'' and
first appeared in the work of Duff and Lu \cite{dufflu} already in~1990. Hamilton densities for membranes have also been discussed around that time, see e.g. \cite{Bars}.\footnote{We believe that the Hamiltonian (\ref{def_polyakovham}) has been known, in this or a similar form, to experts for a long time but we were not able to trace it in even older literature, cf. \cite{Duff:1989tf} for the string case. More recently, the Hamiltonian as well as the open membrane metrics appeared, e.g., in \cite{Berman:2010is}. We thank D. Berman for bringing this paper to our attention.} The block matrix in the Hamiltonian can be viewed as  positive
definite fibrewise metric $\mathbf{G}$ on $T^{\ast}M \oplus \TM{p}$
defined on sections as

\begin{equation}\label{GenMet}
\mathbf{G}(\alpha + \bi{Q}, \beta + \bi{R}) = \bv{\alpha}{\bi{Q}}^{T}
\bm{g^{-1}}{-g^{-1}C}{-C^{T}g^{-1}}{\~g + C^{T}g^{-1}C}
\bv{\beta}{\bi{R}},
\end{equation}
for all $\alpha, \beta \in \df{1}$ and $\bi{Q},\bi{R} \in \vf{p}$. For $p=1$
and $\~g = g$, one gets exactly the inverse of the generalized metric corresponding
to a Riemannian metric $g$ and a $2$-form $C$. Note that, analogously to the $p=1$ case,
$\mathbf{G}$ can be written as a product of  block lower triangular,
diagonal and upper triangular matrices:
\begin{equation} \label{eq_gCdecomposition}
\mathbf{G} = \bm{1}{0}{-C^{T}}{1} \bm{g^{-1}}{0}{0}{\~g} \bm{1}{-C}{0}{1}.
\end{equation}

Before we proceed with our discussion of the corresponding Nambu
sigma models, let us introduce another parametrization of the
background fields $g$ and $C$. In analogy with the $p=1$ case, we
shall refer to $g$ and $C$ as to the closed background fields. Let
$\mathbf{A}$ denote the matrix in the action
(\ref{def_actiontotmatrix}), that is,
\begin{equation}\label{A}
\mathbf{A} = \bm{g}{C}{-C^{T}}{\~g}.
\end{equation}
This matrix is always invertible, explicitly:
\begin{equation} \label{eq_Ainverse}
\mathbf{A}^{-1} = \bm{(g + C\~g^{-1}C^{T})^{-1}}{-(g +
C\~g^{-1}C^{T})^{-1}C\~g^{-1}}{\~g^{-1}C^{T}(g +
C\~g^{-1}C^{T})^{-1}}{(\~g + C^{T}g^{-1}C)^{-1}}.
\end{equation}
Further, let us assume an arbitrary but fixed $(p+1)$-vector $\Pi
\in \mathfrak{X}^{p+1}(M)$ and consider a matrix $\mathbf{B}$ of the form

\begin{equation}
\begin{split} \label{def_Bmatrix}
\mathbf{B}   & = \bm{G}{\Phi}{-\Phi^{T}}{\~G}^{-1} + \bm{0}{\Pi}{-\Pi^{T}}{0} \\
             & = \bm{(G+\Phi \~G^{-1} \Phi^{T})^{-1}}{-(G+\Phi \~G \Phi^{T})^{-1}
             \Phi \~G^{-1} + \Pi}{\~G^{-1}\Phi^{T} (G + \Phi \~G^{-1} \Phi^{T})^{-1} -
             \Pi^{T}}{(\~G + \Phi^{T}G^{-1}\Phi)^{-1}}
\end{split}
\end{equation}
such that the equality $\mathbf{A}^{-1}=\mathbf{B}$, i.e.,
\begin{equation} \label{eq_occorrespondece0}
\bm{g}{C}{-C^{T}}{\~g}^{-1}=\bm{G}{\Phi}{-\Phi^{T}}{\~G}^{-1} +
\bm{0}{\Pi}{-\Pi^{T}}{0}
\end{equation}
holds. This generalization was introduced and used in \cite{Jurco:2012yv}. Again, in analogy with the case $p=1$, we will refer to $G$
and $\Phi$ as to the open backgrounds. More explicitly, we have the
following set of open-closed relations:
\begin{equation} \label{eq_occorrespondence1}
g + C\~g^{-1}C^{T} = G + \Phi \~G^{-1} \Phi^{T},
\end{equation}
\begin{equation} \label{eq_occorrespondence2}
\~g + C^{T} g^{-1} C = \~G + \Phi^{T} G^{-1} \Phi,
\end{equation}
\begin{equation} \label{eq_occorrespondence3}
g^{-1}C = G^{-1}\Phi - \Pi(\~G + \Phi^{T}G^{-1}\Phi),
\end{equation}
\begin{equation} \label{eq_occorrespondence4}
\Phi \~G^{-1} = C\~g^{-1} + (g + C\~g^{-1}C^{T})\Pi.
\end{equation}

For fixed $\Pi$, given $(g,\~g,C)$ there exist unique
$(G,\~G,\Phi)$ such that the above relations are fulfilled, and vice
versa. The explicit expressions are most directly seen from the
equality $\mathbf{A}= \mathbf{B}^{-1}$, again using the formula for
the inverse of the block matrix $\mathbf{B}$. In particular,

\begin{equation}
g^{-1} = (1 - \Phi \Pi^{T})^{T} G^{-1} (1 - \Phi \Pi^{T}) +
\Pi \~G \Pi^{T},
\end{equation}
\begin{equation}
\~g^{-1} = (1 - \Phi^{T} \Pi)^{T} \~G^{-1} (1 - \Phi^{T} \Pi)
+ \Pi^{T} G \Pi,
\end{equation}
and the explicit expression for $C$ can be found straightforwardly.
Obviously, the inverse relations are obtained simply by
interchanging $g \leftrightarrow G$,  $\~g \leftrightarrow \~G$, $C
\leftrightarrow \Phi$, and $\Pi \leftrightarrow -\Pi$.
Using these relations, we can write the action
(\ref{def_actiontotmatrix}) equivalently in terms of the open
backgrounds $G$, $\Phi$ and the (so far auxiliary) $(p+1)$-vector
$\Pi$.

In terms of the corresponding Hamiltonian (\ref{def_polyakovham}),
the above open-closed relations give just another factorization of
the matrix $\mathbf{G}$. This time we have
\begin{equation} \label{def_Gprime}
\mathbf{G} = \bm{1}{\Pi}{0}{1} \bm{1}{0}{-\Phi^{T}}{1}
\bm{G^{-1}}{0}{0}{\~G} \bm{1}{-\Phi}{0}{1} \bm{1}{0}{\Pi^{T}}{1}.
\end{equation}

In the sequel it will be convenient to distinguish the respective
expressions of  above introduced matrices $\mathbf{A}$ and
$\mathbf{G}$ in the closed and open variables. For the former we we
shall use $\mathbf{A_c}$ and $\mathbf{G_c}$ and for the latter we
introduce $\mathbf{A}_o$ and $\mathbf{G}_o$, respectively. Hence the
open-closed relations can be expressed either way:
$\mathbf{A} \equiv \mathbf{A}_c = \mathbf{A}_o \equiv \mathbf{B}^{-1}$ or
$\mathbf{G}_c=\mathbf{G}_o$. Note, that the latter form is just
equivalent to the statement about the decomposability of a 2x2 block
matrix with the invertible upper left block as a product of lower
triangular, diagonal, and upper triangular block matrices, the
triangular ones having unit matrices on the diagonal. Note that for
$p=1$ and $\~g = g$, the open-closed relations (see \cite{Seiberg:1999vs}) are usually written
simply as
\begin{equation} \label{eq_p1occorrespondence}
\frac{1}{g+C} = \frac{1}{G+\Phi} + \Pi.
\end{equation}
To conclude this section, note that taking the determinant of the matrix $\mathbf{A}_{c}$ , we may prove the useful identity:
\begin{equation} \label{eq_detidentity}
\det{(\~g + C^{T} g^{-1} C)} = \frac{\det{\~g}}{\det{g}} \det{(g + C\~g^{-1}C^{T})}.
\end{equation}
To show this, just note that $\mathbf{A}_{c}$ can be decomposed in two different ways, either
\[ \mathbf{A}_{c} = \bm{1}{0}{-C^{T}g^{-1}}{1} \bm{g}{0}{0}{(\~g + C^{T}g^{-1}C)} \bm{1}{g^{-1}C}{0}{1}, \]
or as
\[ \mathbf{A}_{c} = \bm{1}{C\~g^{-1}}{0}{1} \bm{(g+C\~g^{-1}C^{T})}{0}{0}{\~g} \bm{1}{0}{-\~g^{-1}C^{T}}{1}. \]
Taking the determinant of both expressions and comparing them yields (\ref{eq_detidentity}).
\section{Nambu sigma model} \label{sec_nambu}
In analogy with the $p=1$ case, we may ask whether there is a Nambu
sigma model classically equivalent to the action
(\ref{def_actiontotmatrix}). To see this, introduce new auxiliary
fields $\eta_{i}$ and $\~\eta_{J}$, which transform according to
their index structure under a change of coordinates on $M$. Define an
$(n + \binom{n}{p})$-row vector $\Upsilon =
\bv{i \eta_{i}}{\~\eta_{J}}$. The corresponding (non-topological)
Nambu sigma model then has the form:
\begin{equation} \label{def_actionnsm}
S_{NSM}[X,\eta,\~\eta] = -\int d^{p+1}\sigma \big\{ \frac{1}{2}
\Upsilon^{\dagger} \mathbf{A}^{-1} \Upsilon + \Upsilon^{\dagger} \Psi \big\},
\end{equation}
where $\mathbf{A}$ can be either of $\mathbf{A_o}$ and
$\mathbf{A_c}$, supposing that the open-closed relations
$\mathbf{A_o}=\mathbf{A_c}$ hold. Using the equations of motion for
$\Upsilon$, one gets back the Polyakov action
(\ref{def_actiontotmatrix}). For the detailed treatment of Nambu sigma models see \cite{Jurco:2012gc}.

Yet another parametrization of $\mathbf{A}^{-1}$ -- using new
background fields $G_{\NB},\~G_{\NB},\Pi_{\NB}$, which we refer to as Nambu background
fields\footnote{Here, instead of fixing $\Pi$ and finding open
variables in terms of closed ones, we fix $\Phi$ to be zero and
find, again using the open-closed relations, unique $G_{N},\~G_{N},\Pi_{N}$ as
functions of $\g, \~g$ and $C$, or vice versa.} -- can be introduced
\begin{equation} \label{def_Ainvnewfield}
\mathbf{A}^{-1} = \bm{G_{\NB}^{-1}}{\Pi_{\NB}}{-\Pi_{\NB}^{T}}{\~G_{\NB}^{-1}}.
\end{equation}
We will denote as $\mathbf{A}_N$ the matrix $\mathbf{A}$ expressed
with help of Nambu background fields $G_{\NB},\~G_{\NB},\Pi_{\NB}$. Using
(\ref{eq_Ainverse}), one gets the correspondence between closed and
Nambu sigma background fields:

\begin{equation} \label{eq_correspondence1}
G_{\NB} = g + C\~g^{-1}C^{T},
\end{equation}
\begin{equation} \label{eq_correspondence2}
\~G_{\NB} = \~g + C^{T}g^{-1}C,
\end{equation}
\begin{equation} \label{eq_correspondence3}
\Pi_{\NB} = -(g + C\~g^{-1}C^{T})^{-1}C\~g^{-1} = -g^{-1}C(\~g + C^{T} g^{-1}C)^{-1}.
\end{equation}

Clearly, $G_{\NB}$ is a Riemannian metric on $M$ and $\~G_{\NB}$ is a fibrewise
positive  definite metric on $\TM{p}$. It is important to note that
 in general, for $p>1$, $\Pi_{\NB}: \cTM{p} \rightarrow TM$ is not
necessarily induced by a $(p+1)$-vector on $M$. This also means that it is not in general a Nambu-Poisson tensor. However; for $p=1$, it is easy to show that $\Pi_{N}$ is a bivector.

Also note that even if $\~g$ is a skew-symmetrized tensor product of
$g$'s  (\ref{def_tensorofg}), $\~G_{\NB}$ is not in general the
skew-symmetrized tensor product of $G_N$'s.

The converse relations are:
\begin{equation}
g = (G_{\NB}^{-1} + \Pi_{\NB} \~G_{\NB} \Pi_{\NB}^{T})^{-1},
\end{equation}
\begin{equation}
\~g = (\~G_{\NB}^{-1} + \Pi_{\NB}^{T} G_{\NB} \Pi_{\NB})^{-1},
\end{equation}
\begin{equation}
C = -(G_{\NB}^{-1} + \Pi_{\NB} \~G_{\NB} \Pi_{\NB}^{T})^{-1} \Pi_{\NB} \~G_{\NB} = -G_{\NB} \Pi_{\NB} (\~G_{\NB}^{-1} +
\Pi_{\NB}^{T} G_{\NB} \Pi_{\NB})^{-1}.
\end{equation}

Again, it is instructive to pass to the corresponding Hamiltonians.
First, find the canonical Hamiltonian to (\ref{def_actionnsm}), that
is
\[ H^{c}_{NSM}[X,P,\~\eta] = \int d^{p}\sigma P_{i} \partial_{0}X^{i} - \mathcal{L}[X,P,\~\eta]. \]
Second, use the equations of motion to get rid of $\~\eta$. In
analogy with the $p=1$ case, one expects that resulting  Hamiltonian
$H_{NSM}$ coincides with (\ref{def_polyakovham}), that is
\[ H_{NSM}[X,P] = H_{P}^{tot}[X,P].\] Indeed, we get
\begin{equation} \label{eq_nsmham}
H_{NSM}[X,P] = -\frac{1}{2} \int d^{p}\sigma \bv{iP}{\px{}}^{T} \bm{G_{\NB}^{-1} + \Pi_{\NB} \~G_{\NB} \Pi_{\NB}^{T}}{\Pi_{\NB} \~G_{\NB}}{\~G_{\NB} \Pi_{\NB}^{T}}{	\~G_{\NB}} \bv{iP}{\px{}}.
\end{equation}
If one plugs (\ref{eq_correspondence1} - \ref{eq_correspondence2})
to (\ref{eq_nsmham}), one obtains exactly the Hamiltonian
(\ref{def_polyakovham}). The matrix $\mathbf{G}$ can be thus written
as
\begin{equation} \label{def_Gdualfields}
\mathbf{G} = \bm{1}{\Pi_{\NB}}{0}{1} \bm{G_{\NB}^{-1}}{0}{0}{\~G_{\NB}}
\bm{1}{0}{\Pi_{\NB}^{T}}{1}
\end{equation}
when using the Nambu background fields, in which case we shall
introduce the notation $\mathbf{G}_\NB$ for it. This shows that to any $g,\~g,C$ one can uniquely find $G_{\NB},\~G_{\NB},\Pi_{\NB}$ and
vice versa, since they both come from the respective unique
decompositions of the matrix $\mathbf{G}$.

Note that for $p=1$ and $\~g = g$, relations
(\ref{eq_correspondence1} - \ref{eq_correspondence3})  are usually
written simply as
\begin{equation} \label{eq_p1correspondence}
\frac{1}{g+C} = \frac{1}{G_{\NB}} + \Pi_{\NB}.
\end{equation}



We will refer to the Poisson sigma model, when expressed -- using
$\Pi$ -- in open variables $(G,\~G,\Phi)$ as to augmented Poisson sigma model.
\section{Geometry of the open-closed brane relations} \label{sec_geometry}
For $p=1$, the open-closed relations (\ref{eq_p1occorrespondence}) can naturally be explained using the language of generalized geometry. We have developed this point of view in \cite{Jurco:2013upa}. One expects that similar observations apply also for $p>1$ case. In the previous section we have already mentioned the possibility to define the generalized metric on the vector bundle $TM \oplus \cTM{p}$ by the inverse of the matrix \eqref{eq_gCdecomposition}. Here we discuss an another approach to a generalization of the generalized geometry starting from equation (\ref{eq_occorrespondece0}). Denote $W = TM \oplus \TM{p}$.

The main goal of this section is to show that we can without any additional labor adapt the whole formalism of \cite{Jurco:2013upa} to the vector bundle $W \oplus W^{\ast}$.

Define the maps $\G$, $\B: W \rightarrow W^{\ast}$ using block matrices as
\begin{equation} \label{eq_bigGbigB}
 \G \cv{V}{\bi{P}} = \bm{g}{0}{0}{\~g} \cv{V}{\bi{P}}, \ \ \B \cv{V}{\bi{P}} = \bm{0}{C}{-C^{T}}{0} \cv{V}{\bi{P}},
\end{equation}
for all $V + \bi{P} \in \Gamma(W)$. Next, define the map $\Theta: W^{\ast} \rightarrow W$ as
\begin{equation} \label{eq_Pbigdef}
 \Theta \cv{\alpha}{\Sigma} = \bm{0}{\Pi}{-\Pi^{T}}{0} \cv{\alpha}{\Sigma},
\end{equation}
for all $\alpha + \Sigma \in \Gamma(W^{\ast})$. Then define $\H,\Xi: W \rightarrow W^{\ast}$ as in (\ref{eq_bigGbigB}) using the fields $G,\~G,\Phi$ instead of $g,\~g,C$. The open-closed relations (\ref{eq_occorrespondece0}) can be then written as simply as
\begin{equation} \label{eq_ocalaSW}
 \frac{1}{\G + \B} = \frac{1}{\H + \Xi} + \Theta.
\end{equation}
We see that they have exactly the same form as (\ref{eq_p1occorrespondence}) for $p=1$. The purpose of this section is to obtain these relations from the geometry of the vector bundle $W \oplus W^{\ast}$.

We define an inner product $\<\cdot,\cdot\>: \Gamma(W \oplus W^{\ast}) \times \Gamma(W \oplus W^{\ast}) \rightarrow C^{\infty}(M)$ on $W \oplus W^{\ast}$ to be the natural pairing between $W$ and $W^{\ast}$, that is:
\[ \< V + \bi{P} + \alpha + \Sigma, W + \bi{Q} + \beta + \Psi \> = \beta(V) + \alpha(W) + \Psi(\bi{P}) + \Sigma(\bi{Q}), \]
for all $V,W \in \vf{}$, $\alpha,\beta \in \df{1}$, $\bi{P},\bi{Q} \in \vf{p}$, and $\Sigma,\Psi \in \df{p}$. Note that this pairing has the signature $(n+\binom{n}{p},n+\binom{n}{p})$.

Now, let $\T: W \oplus W^{\ast} \rightarrow W \oplus W^{\ast}$ be a vector bundle endomorphism squaring to identity, that is, $\T^{2} = 1$. We say that $\T$ is a generalized metric on $W \oplus W^{\ast}$, if the fibrewise bilinear form
\[ (E_{1},E_{2})_{\T} \equiv \< E_{1}, \T(E_{2})\>, \]
defined for all $E_{1},E_{2} \in \Gamma(W \oplus W^{\ast})$, is a positive definite fibrewise metric on $W \oplus W^{\ast}$. It follows from definition that $\T$ is orthogonal and symmetric with respect to the inner product $\<\cdot,\cdot\>$. Moreover, it defines two eigenbundles $V_{\pm} \subset W \oplus W^{\ast}$, corresponding to eigenvalues $\pm 1$ of $\T$. It follows immediately from the properties of $\T$, that they are both of rank $n + \binom{n}{p}$, orthogonal to each other, and thus
\[ W \oplus W^{\ast} = V_{+} \oplus V_{-}. \]
Moreover, $V_{+}$ and $V_{-}$ form the positive definite and negative definite subbundles of $\<\cdot,\cdot\>$, respectively. From the positive definiteness of $V_{+}$ it follows that $V_{+}$ has zero intersection both with $W$ and $W^{\ast}$, and is thus a graph of a unique vector bundle isomorphism $\mathcal{A}: W \rightarrow W^{\ast}$. The map $\mathcal{A}$ can be written as a sum of a symmetric and a skew-symmetric part with respect to $\<\cdot,\cdot\>$: $\mathcal{A} = \G + \B$. From the positive definiteness of $V_{+}$, it follows that $\G$ is a positive definite fibrewise metric on $W$. From the orthogonality of $V_{+}$ and $V_{-}$ we finally obtain that:
\[ V_{\pm} = \{ (V + \bi{P}) + ( \pm \G + \B )(V + \bi{P}) \ | V+ \bi{P} \in W \}. \]
The map $\T$, or equivalently the fibrewise metric $(\cdot,\cdot)_{\T}$ can be reconstructed using the data $\G$ and $\B$ to get
\[ ( V + \bi{P} + \alpha + \Sigma, W + \bi{Q} + \beta + \Psi)_{\T} = \cv{V + \bi{P}}{\alpha + \Sigma}^{T} \bm{ \G - \B \G^{-1} \B }{\B \G^{-1}}{-\G^{-1} \B}{\G^{-1}} \cv{W+\bi{Q}}{\beta + \Psi}. \]
Note that the above block matrix can be decomposed as a product
\[ \bm{ \G - \B \G^{-1} \B }{\B \G^{-1}}{-\G^{-1} \B}{\G^{-1}} = \bm{1}{\B}{0}{1} \bm{\G}{0}{0}{\G^{-1}} \bm{1}{0}{-\B}{1}. \]
The maps $\G,\B$ can be parametrized as
\[ \G \cv{V}{\bi{Q}} = \bm{g}{D}{D^{T}}{\~g} \cv{V}{\bi{Q}}, \]
\[ \B \cv{V}{\bi{Q}} = \bm{B}{C}{-C^{T}}{\~B} \cv{V}{\bi{Q}}, \]
where $g$ is a symmetric covariant $2$-tensor on $M$, $C,D: \TM{p} \rightarrow T^{\ast}M$ are vector bundle morphisms, $B \in \Omega^{2}(M)$, and $\~g$ and $\~B$ are symmetric and skew-symmetric fibrewise bilinear forms on $\TM{p}$, respectively. The fields $g,\~g,D$ are not arbitrary, since $\G$ has to be a positive definite fibrewise metric on $W$. One immediately gets that $g,\~g$ have to be positive definite. The conditions imposed on $D$ can be seen from the equalities
\[
\begin{split}
\bm{g}{D}{D^{T}}{\~g} & = \bm{1}{0}{D^{T}g^{-1}}{1} \bm{g}{0}{0}{\~g - D^{T}g^{-1}D} \bm{1}{g^{-1}D}{0}{1} \\
& = \bm{1}{D\~g^{-1}}{0}{1} \bm{g - D\~g^{-1}D^{T}}{0}{0}{\~g} \bm{1}{0}{\~g^{-1}D^{T}}{1}.
\end{split}
\]
We see that there are two equivalent conditions on $D$: the fibrewise bilinear form $\~g - D^{T}g^{-1}D$, or $2$-tensor $g - D\~g^{-1}D^{T}$ have to be positive definite. Inspecting the action (\ref{def_actiontotmatrix}), we see that only the case when $B = \~B = D = 0$ is relevant for our purpose.

Now, let us turn our attention to the explanation of the open-closed relations. For this, consider the vector bundle automorphism $\mathcal{O}: W \oplus W^{\ast} \rightarrow W \oplus W^{\ast}$, orthogonal with respect to the inner product $\<\cdot,\cdot\>$, that is,
\[ \< \mathcal{O}(E_{1}), \mathcal{O}(E_{2}) \> = \< E_{1}, E_{2} \>, \]
for all $E_{1},E_{2} \in \Gamma(W \oplus W^{\ast})$. Given a generalized metric $\T$, we can define a new map $\T' = \mathcal{O}^{-1} \T \mathcal{O}$. It can be easily checked that $\T'$ is again a generalized metric. Obviously, the respective eigenbundles $V_{+}$ are related using $\mathcal{O}$, namely:
\begin{equation} \label{eq_relgraphs}
 V_{+}^{\T'} = \mathcal{O}^{-1} ( V_{+}^{\T} ).
\end{equation}
We have also proved that every generalized metric $\T$ corresponds to two unique fields $\G$ and $\B$. This means that to given $\G$ and $\B$, and an orthogonal vector bundle isomorphism $\mathcal{O}$, there exists a unique pair $\H$, $\Xi$ corresponding to $\T' = \mathcal{O}^{-1} \T \mathcal{O}$. We will show that open-closed relations are a special case of this correspondence. Also, note that $(\cdot,\cdot)_{\T}$ and $(\cdot,\cdot)_{\T'}$ are related as
\begin{equation} \label{eq_GMogtransform}
 (\cdot,\cdot)_{\T'} = (\mathcal{O}(\cdot),\mathcal{O}(\cdot))_{\T}.
\end{equation}
Now, consider an arbitrary skew-symmetric morphism $\Theta: W^{\ast} \rightarrow W$, that is
\[ \< \alpha + \Sigma, \Theta(\beta + \Psi) \> = - \< \Theta(\alpha + \Sigma), \beta + \Psi \>, \]
for all $\alpha, \beta \in \df{1}$, and $\Sigma,\Psi \in \df{p}$. It can easily be seen that the vector bundle isomorphism $e^{\Theta}: W \oplus W^{\ast} \rightarrow W \oplus W^{\ast}$, defined as
\[ e^{\Theta} \cv{ V+\bi{Q}}{\alpha + \Sigma} = \bm{1}{\Theta}{0}{1} \cv{V + \bi{Q}}{\alpha + \Sigma}, \]
for all $V + \bi{Q} + \alpha + \Sigma \in \Gamma(W \oplus W^{\ast})$, is orthogonal with respect to the inner product $\<\cdot,\cdot\>$. Its inverse is simply $e^{-\Theta}$. Let $\T$ be the generalized metric corresponding to $\G + \B$. Note that $V_{+}^{\T}$ can be expressed as
\[ V_{+}^{\T} = \{ (\G + \B)^{-1}(\alpha + \Sigma) + (\alpha + \Sigma) \ | \ (\alpha + \Sigma) \in W^{\ast} \}. \]
Using the relation (\ref{eq_relgraphs}), we obtain that
\[ V_{+}^{\T'} = e^{-\Theta} V_{+}^{\T} = \{ \big( (\G + \B)^{-1} - \Theta \big)(\alpha + \Sigma) + (\alpha + \Sigma) \ | \ (\alpha + \Sigma) \in W^{\ast} \}. \]
We see that the vector bundle morphism $\H + \Xi$ corresponding to $\T'$ satisfies
\[ (\H + \Xi)^{-1} = (\G + \B)^{-1}  - \Theta. \]
But this is precisely the relation (\ref{eq_ocalaSW}). We also know how to handle this relation on the level of the positive definite fibrewise metrics $(\cdot,\cdot)_{\tau}$ and $(\cdot,\cdot)_{\tau'}$. From (\ref{eq_GMogtransform}) we get the relation
\[ \bm{\H - \Xi \H^{-1} \Xi}{\B \H^{-1}}{-\H^{-1} \Xi}{\H^{-1}} = \bm{1}{0}{-\Theta}{1} \bm{\G - \B \G^{-1} \B}{\B \G^{-1}}{-\G^{-1} \B}{\G^{-1}} \bm{1}{\Theta}{0}{1}. \]
Using the decomposition of the matrices, we can write this also as
\[ \bm{1}{\Xi}{0}{1} \bm{\H}{0}{0}{\H^{-1}} \bm{1}{0}{-\Xi}{0} = \bm{1}{0}{-\Theta}{1}  \bm{1}{B}{0}{1} \bm{\G}{0}{0}{\G^{-1}} \bm{1}{0}{-B}{1} \bm{1}{\Theta}{0}{1}. \]
Comparing both expressions, we get the explicit form of open-closed relations:
\begin{equation} \label{eq_ocrelationsbig1}
\H - \Xi \H^{-1} \Xi = \G - \B \G^{-1} \B,
\end{equation}
\begin{equation} \label{eq_ocrelationsbig2}
\Xi \H^{-1} = (\G - \B \G^{-1}\B)\Theta + \B\G^{-1},
\end{equation}
\begin{equation} \label{eq_ocrelationsbig3}
\H^{-1} = (1 + \Theta \B)\G^{-1}(1 - \B \Theta) - \Theta \G \Theta.
\end{equation}
We have proved that for given $\G,\B$ and any $\Theta$, $\H$ and $\Xi$ can be found uniquely. Inverse relations can be obtained by interchanging $\G \leftrightarrow \H$, $\B \leftrightarrow \Xi$ and $\Theta \leftrightarrow -\Theta$. Note that, actually, the last equation follows from the first two. Now let us turn our attention to the case of $\G + \B$ in the form (\ref{eq_bigGbigB}). One has
\[ \G - \B \G^{-1} \B = \bm{g + C\~g^{-1}C^{T}}{0}{0}{\~g+C^{T}g^{-1}C}, \]
\[ \B \G^{-1} = \bm{0}{C\~g^{-1}}{-C^{T}g^{-1}}{0}, \ \ \G^{-1} = \bm{g^{-1}}{0}{0}{\~g^{-1}}. \]
Parametrize $\Theta$ as
\[ \Theta = \bm{\pi}{\Pi}{-\Pi^{T}}{\~\pi}, \]
where $\pi \in \vf{2}$, $\Pi: \cTM{p} \rightarrow TM$, and $\~\pi$ is skew-symmetric fibrewise bilinear form on $\cTM{p}$. Right-hand side of (\ref{eq_ocrelationsbig2}) is then
\[
\begin{split}
\bm{g+C\~g^{-1}C^{T}}{0}{0}{\~g+C^{T}g^{-1}C} \bm{\pi}{\Pi}{-\Pi^{T}}{\~\pi} + \bm{0}{C\~g^{-1}}{-C^{T}g^{-1}}{0} = \\
= \bm{ (g + C\~g^{-1}C^{T})\pi }{(g+C\~g^{-1}C^{T})\Pi + C\~g^{-1}}{-(\~g + C^{T}g^{-1}C)\Pi^{T} - C^{T}g^{-1}}{(\~g + C^{T}g^{-1}C)\~\pi  }.
\end{split}
\]
We see that to obtain a generalized metric where $\mathcal{H}$ is block diagonal, and $\Xi$ is block off-diagonal, we have to choose $\pi = \~\pi = 0$. This means that we choose $\Theta$ to be of the form (\ref{eq_Pbigdef}). Defining
\[ \H = \bm{G}{0}{0}{\~G}, \ \Xi = \bm{0}{\Phi}{-\Phi^{T}}{0}, \]
it is now straightforward to see that the set of equations (\ref{eq_ocrelationsbig1} - \ref{eq_ocrelationsbig3}) gives exactly the open-closed relations (\ref{eq_occorrespondence1} - \ref{eq_occorrespondence4}). The relations between the open membrane variables and Nambu fields $G_{\NB},\~G_{\NB},\Pi_{\NB}$ can be explained in a similar fashion. Indeed, note that the map $\G + \B$ is invertible, and its inverse, the vector bundle morphism from $W^{\ast}$ to $W$, can be split into symmetric and skew-symmetric part:
\begin{equation} \label{eq_openNambubig}
 (\G + \B)^{-1} = \mathcal{H}_{\NB}^{-1} + \Theta_{\NB},
\end{equation}
where $\mathcal{H}_{\NB}$ is a fibrewise positive definite metric on $W$, and $\Theta_{\NB}$ is a skew-symmetric fibrewise bilinear form on $W^{\ast}$. Parametrizing them as
\[ \mathcal{H}_{\NB} = \bm{G_{\NB}}{0}{0}{\~G_{\NB}}, \ \ \Theta_{\NB} = \bm{0}{\Pi_{\NB}}{-\Pi_{\NB}^{T}}{0}, \]
and expanding (\ref{eq_openNambubig}), we obtain exactly the set of equations (\ref{eq_correspondence1} - \ref{eq_correspondence3}).
\section{Gauge field $F$ as transformation of the fibrewise metric} \label{sec_gaugeF}
In this section, we would like to develop the equalities required in the discussion of DBI actions. In the previous sections we have shown how the closed and open membrane actions are related using the generalized geometry point of view. One expects that it is also true for their versions taking into account the fluctuations. The following paragraphs show that it is true ``up to an isomorphism", fluctuated backgrounds cannot be related simply by open-closed relations in the form (\ref{eq_occorrespondence1} - \ref{eq_occorrespondence4}).

We also show that corresponding open backgrounds are essentially uniquely fixed, there is no ambiguity at all. For $p=1$, we have already used this observation in \cite{Jurco:2013upa}.

The idea is the following: Suppose that we would like to add a fluctuation $F$ to the
$(p+1)$-form $C$. At this point we consider  $F$ to be
defined globally on the entire manifold $M$, although everything works also in the case when
$F$ is defined only on a some submanifold of $M$.\footnote{Later,
this submanifold will correspond to a $p'$-brane, $p'\geq p$, where
$p$-branes can end.}

Going from $C$ to $C+F$ corresponds to
replacing $\mathbf{G}_c$ in the Hamiltonian (\ref{def_polyakovham})
with $\mathbf{G}_c^{F}$, defined as
\begin{equation} \label{def_GFmatrix}
\mathbf{G}_c^{F} = \bm{1}{0}{-F^{T}}{1} \mathbf{G}_c
\bm{1}{-F}{0}{1} \equiv \bm{1}{0}{-(C+F)^{T}}{1}
\bm{g^{-1}}{0}{0}{\~g} \bm{1}{-(C+F)}{0}{1}.
\end{equation}
The matrix $\bm{1}{-F}{0}{1}$ corresponds to an endomorphism of $T^{\ast}M
\oplus \TM{p}$, which we denote as $e^{-F}$. Note that unlike in the
$p=1$ case, $e^{-F}$ is not orthogonal with respect to the canonical
pairing (valued in $\vf{p-1}$) on $T^{\ast}M \oplus \TM{p}$, defined
as:
\[ \<\alpha+\bi{Q},\beta+ \bi{R}\> = \io_{\alpha}\bi{R} + \io_{\beta}\bi{Q}, \]
for all $\alpha,\beta \in \df{1}$ and $\bi{Q},\bi{R} \in \vf{p}$. It can be
shown that any orthogonal $F$ has to be identically $0$. On the other
hand, its transpose map, $(e^{-F})^{T} \equiv e_{-F}$, which is an
endomorphism of $TM \oplus \cTM{p}$, is orthogonal with respect to
the canonical pairing (valued in $\df{p-1}$) on $TM \oplus \cTM{p}$
iff $F$ is a $(p+1)$-form in $M$. This pairing is defined as
\[ \<V + \Sigma, W + \Xi\> = \io_{V}\Sigma + \io_{W}\Xi,  \]
for all $V,W \in \vf{}$ and $\Sigma,\Xi \in \df{p}$. In this notation, the transformation
(\ref{def_GFmatrix}) can be written   as
\begin{equation}  \label{eq_Gffactorization}
\mathbf{G}_c^{F} = e_{-F} \mathbf{G}_c e^{-F} \equiv (e^{-F})^{T}
\mathbf{G}_c e^{-F}.
\end{equation}
We know that $\mathbf{G}$ can be rewritten as $\mathbf{G}_o$ in the open
variables $(G,\~G,\Phi)$,  corresponding to augmented Nambu sigma
model. If we define the automorphism $e^{\Pi}$ of $T^{\ast}M
\oplus \TM{p}$ as
\[ e^{\Pi} \bv{\alpha}{\bi{Q}} = \bm{1}{0}{\Pi^{T}}{1} \bv{\alpha}{\bi{Q}}, \]
we can express $\mathbf{G}_o$ as
\begin{equation}
\mathbf{G}_o = e_{\Pi} \bm{1}{0}{-\Phi^{T}}{1}
\bm{G^{-1}}{0}{0}{\~G}  \bm{1}{-\Phi}{0}{1} e^{\Pi},
\end{equation}
where $e_{\Pi} = (e^{\Pi})^{T}$. Dually to the previous
discussion, $e^{\Pi}$ is an orthogonal transformation of
$T^{\ast}M \oplus \TM{p}$; although $e_{\Pi}$, for non-zero
$\Pi$, is never orthogonal on $TM \oplus \cTM{p}$.

Now, it is natural to ask whether to the gauged closed variables
$(g,\~g,C+F)$ there correspond some open variables and hence an
augmented Nambu sigma model, described by  some $\Pi'$ and
$(G,\~G,\Phi + F')$, where $F'$ describes a fluctuation of the background
$\Phi$. More precisely, we ask whether one can write
$\mathbf{G}_o^{F}$ in the form
\begin{equation}
\mathbf{G}_o^{F} \stackrel{?}{=} e_{\Pi'} \bm{1}{0}{-(\Phi +
F')^{T}}{1} \bm{G^{-1}}{0}{0}{\~G} \bm{1}{-(\Phi + F')}{0}{1}
e^{\Pi'}.
\end{equation}
Translated into the language of the corresponding automorphisms of $T^{\ast}M \oplus
\TM{p}$, this boils down to the question
\begin{equation} \label{eq_ortequality}
e^{\Pi} e^{-F}\stackrel{?}{=} e^{-F'} e^{\Pi'},
\end{equation}
for some $\Pi'$ and $F'$. In general, this is not possible. Explicitly the
equation (\ref{eq_ortequality}) reads
\[ \bm{1}{-F}{\Pi^{T}}{1 - \Pi^{T}F} \stackrel{?}{=}
\bm{1 - F'\Pi'}{-F'}{\Pi'^{T}}{1}. \] This implies $\Pi^{T}
F = 0$, which, of course, in general is not satisfied. The
decomposition on the right-hand side therefore has to contain a
block-diagonal term. Note that $e^{-F'}$ is upper triangular,
whereas $e^{\Pi'}$ is lower triangular. For a matrix to have a
decomposition into a product of a block upper triangular, diagonal
and lower triangular matrix, it has to have an invertible bottom
right block, that is $1 - \Pi^{T}F$. Hence, we assume that $1 -
\Pi^{T} F$ is an invertible $\binom{n}{p} \times \binom{n}{p}$
matrix. We are now looking for a solution of the equation
\begin{equation} \label{eq_ortequality2}
 e^{\Pi}e^{-F} = e^{-F'} \bm{M}{0}{0}{N} e^{\Pi'},
\end{equation}
where $M: T^{\ast}M \rightarrow T^{\ast}M$ and $N: \TM{p}
\rightarrow \TM{p}$ are (necessarily) invertible vector bundle
morphisms.

We can decompose $e^{\Pi} e^{-F}$ as
\begin{equation}
\bm{1}{-F(1-\Pi^{T}F)^{-1}}{0}{1}
\bm{1 + F(1 - \Pi^{T}F)^{-1}\Pi^{T}}{0}{0}{1 - \Pi^{T}F}
\bm{1}{0}{(1 - \Pi^{T}F)^{-1}\Pi^{T}}{1}.
\end{equation}
From this we see that $F' = F(1 - \Pi^{T} F)^{-1}$, $\Pi' =
\Pi (1 - F^{T}\Pi)^{-1} $ and $N = 1 - \Pi^{T}F$. To find
an alternative description of $F'$, $\Pi'$ and $M$, examine the
inverse of the equation (\ref{eq_ortequality2}):
\begin{equation}
e^{F} e^{-\Pi} = e^{-\Pi'}  \bm{M^{-1}}{0}{0}{N^{-1}} e^{F'}.
\end{equation}
The left hand side of this equation is
\[ e^{F} e^{-\Pi} = \bm{1 - F\Pi^{T}}{F}{-\Pi^{T}}{1}, \]
which shows that $1 - \Pi^{T}F$ is invertible iff $1 -
F\Pi^{T}$ is invertible. The decomposition of $e^{F}e^{-\Pi}$
reads
\begin{equation}
\bm{1}{0}{-\Pi^{T}(1 - F\Pi^{T})^{-1}}{1}
\bm{1 - F\Pi^{T}}{0}{0}{1 + \Pi^{T}(1 - F \Pi^{T})^{-1}F}
\bm{1}{(1-F\Pi^{T})^{-1}F}{0}{1}.
\end{equation}
We thus get that $F' = (1-F\Pi^{T})^{-1}F$, $\Pi' = (1 -
\Pi F^{T})^{-1} \Pi$ and $M = (1 - F\Pi^{T})^{-1}$.

We can conclude that the fields $F'$, $\Pi'$, and vector bundle
morphisms $M,N$ in the decomposition (\ref{eq_ortequality2}) have
one of the following equivalent forms:

\begin{equation} \label{eq_primed1}
F' = F(1-\Pi^{T}F)^{-1} = (1-F\Pi^{T})^{-1}F,
\end{equation}
\begin{equation} \label{eq_primed2}
\Pi' = \Pi(1 - F^{T}\Pi)^{-1} = (1 - \Pi F^{T})^{-1} \Pi,
\end{equation}
\begin{equation} \label{eq_primed3}
M = 1 + F(1 - \Pi^{T}F)^{-1}\Pi^{T} = 1 + F'\Pi^{T} =
(1 - F\Pi^{T})^{-1},
\end{equation}
\begin{equation} \label{eq_primed4}
N = 1 - \Pi^{T}F = \big( 1 + \Pi^{T}(1 - F\Pi^{T})^{-1}F
\big)^{-1} = (1 + \Pi'^{T}F)^{-1}.
\end{equation}

Thus, we have found a factorization of $\mathbf{G}_o^{F}$ in the
form
\begin{equation}
\mathbf{G}_o^{F} = e_{\Pi'} \bm{M^{T}}{0}{0}{N^{T}} e_{-(\Phi +
F')} \bm{G^{-1}}{0}{0}{\~G} e^{-(\Phi + F')} \bm{M}{0}{0}{N}
e^{\Pi'}.
\end{equation}
Comparing this to  $\mathbf{G}_c^{F}$, in particular comparing
the respective bottom right blocks, we get the important identity
\begin{equation}
\~g + (C+F)^{T}g^{-1}(C+F) = N^{T}\big(\~G + (\Phi + F')^{T}G^{-1}
(\Phi + F')\big)N.
\end{equation}
Similarly, comparing the top left blocks of the inverses, one gets
\begin{equation}\label{openclosedF}
g + (C+F)\~g^{-1}(C+F)^{T} = M^{-1}\big(G + (\Phi + F')\~G^{-1}
(\Phi + F'\big)^{T}M^{-T}.
\end{equation}

Equivalently, one can gauge the matrix $\mathbf{A}_c$, i.e., set
\begin{equation}\label{ACF}
{\mathbf{A}}_c^F=\begin{pmatrix} g & (C+F )\\
-(C+F)^T& \~g\end{pmatrix}.
\end{equation} To express this matrix in
open variables we introduce the following notation: $\bar G^{-1}
:=M^TG^{-1}M$, $\bar{\~G}=N^T\~G N$, $\bar \Phi :=M^{-1}\Phi N$ and
$\bar F' :=M^{-1}F'N$. If we now put
\begin{equation}
{\mathbf{A}}_o^F=\begin{pmatrix} \bar G &  (\bar\Phi + \bar F') \\
-(\bar \Phi +\bar F')^T & \bar{\~G}\end{pmatrix}^{-1}
 +
\begin{pmatrix} 0 & \Pi' \\
-\Pi'^T & 0\end{pmatrix}, \end{equation} the (gauged) open-closed
relations are equivalent to ${\mathbf{A}}_c^F={\mathbf{A}}_o^F$. As in the previous sections, using the matrices ${\mathbf{A}}_c^F$,
${\mathbf{A}}_o^F$, ${\mathbf{G}}_c^F$ and ${\mathbf{G}}_o^F$, one
can write down the corresponding Polyakov or (augmented) Nambu sigma
models, i.e.,
\begin{equation} \label{def_actiontotmatrixF}
S_{P}^{tot,F}[X] = \frac{1}{2} \int d^{p+1}\sigma \{ \Psi^{\dagger} \mathbf{A}^F_c \Psi \}=\frac{1}{2} \int d^{p+1}\sigma \{ \Psi^{\dagger} \mathbf{A}^F_o \Psi \},
\end{equation}
\begin{equation} \label{def_actionNambumatrixF}
S_{NSM}^{F}[X,\eta,\eta'] = - \int d^{p+1}\sigma \{ \Upsilon^{\dagger} {\mathbf{A}^F_c}^{-1} \Upsilon  + \Upsilon^{\dagger} \Psi \}=- \int d^{p+1}\sigma \{ \Upsilon^{\dagger} {\mathbf{A}^F_c}^{-1} \Upsilon + \Upsilon^{\dagger} \Psi \},
\end{equation}
\begin{equation} \label{def_polyakovhamF}
\begin{split}
H_{P}^{tot,F}[X,P] = H_{NSM}^F[X,P] & = -\frac{1}{2} \int d^{p}\sigma \bv{iP}{\px{}}^{T} \mathbf{G_c^F}
\bv{iP}{\px{}} \\
& =  -\frac{1}{2} \int d^{p}\sigma \bv{iP}{\px{}}^{T} \mathbf{G_o^F}
\bv{iP}{\px{}}.
\end{split}
\end{equation}

\section{Seiberg-Witten map} \label{sec_SWmap}
In the previous section, we have developed the correspondence between closed and open fields, including their respective fluctuations. However, they are not related simply by open-closed relations. Instead, the discussion brings new vector bundle isomorphisms $M$ and $N$, defined by (\ref{eq_primed3}, \ref{eq_primed4}), respectively, into the picture. The determinant of the left-hand side of (\ref{openclosedF}) seems to be a likely candidate to appear in the  ``commutative" membrane DBI action, whereas the determinant on the right-hand side of (\ref{openclosedF}) seems to contain as a factor a likely candidate to appear in its ``noncommutative" counterpart.

This observation suggests that we should look for a change of coordinates on the manifold $M$, the Jacobian of which could cancel the $\det{}^{2}(N)$ factor coming under the determinant from the right-hand side of (\ref{openclosedF}).
The resulting diffeomorphism will be called a Seiberg-Witten map in analogy to the string $p=1$ case. We use a direct generalization of the semi-classical construction used first in \cite{Jurco:2001my}. The most intriguing part will be to define carefully a substitute for a determinant of a Nambu-Poisson $(p+1)$-vector.

In the following, let $\Pi$ be a Nambu-Poisson $(p+1)$-vector
(see appendix \ref{sec_NPstructures}) on $M$. We can examine the
$F$-gauged tensor $\Pi' = (1 - \Pi F^{T})^{-1} \Pi$.\footnote{We assume that $1 - \Pi F^{T}$ is invertible. In a more formal approach we also could treat $\Pi'$ as a formal power series in $\Pi$.\label{invertible}}
We will now show that for $p>1$ this tensor is always a Nambu-Poisson $(p+1)$-vector, whereas for
$p=1$ it is a Poisson bivector if $F$ is closed.

First, for $p > 1$, one can see that
\begin{equation}  \label{eq_thetaprimetrace}
\Pi' = (1 - \frac{1}{p+1}\< \Pi, F\>)^{-1} \Pi,
\end{equation}
where $\<\Pi,F\> = \Pi^{iJ} F_{iJ} \equiv Tr(\Pi F^{T})$.
For this, one has to prove that
\begin{equation} \label{TwoTheta_prime_expressions}\Pi = (1 - \Pi F^{T}) (1 - \frac{1}{p+1}\<\Pi,F\>)^{-1} \Pi.
\end{equation}
This can easily be checked in coordinates $(x^{1}, \dots, x^{n})$ in which (\ref{eq_Pidecomposition}) holds, and hence, for $\Pi$ with components $\Pi^{iJ} = \epsilon^{iJ}$. Now, using (\ref{eq_thetaprimetrace}) and lemma
\ref{lem_nptensormultiple}, we see that $\Pi'$ is again a
Nambu-Poisson tensor.

To include the $p=1$ case:
For $p \geq 1$, and $F$ closed, we can use the fact that
$G_{\Pi'} = e_{-F} G_{\Pi}$, where $G_{\Pi}$ and
$G_{\Pi'}$ are graphs of the maps $\Pi^{\sharp}$ and
$\Pi'^{\sharp}$, respectively (see lemma \ref{lem_nplemma}). This is easily
verified using (\ref{eq_primed2}). It can be seen that the Dorfman
bracket (\ref{def_dorfman}) satisfies
$[e_{-F}(V+\xi),e_{-F}(W+\eta)]_{D} = e_{-F}[V+\xi,W+\eta]_{D}$,
whenever $F$ is closed. But this implies that $G_{\Pi'}$ is
closed under the Dorfman bracket, which is according to
\ref{lem_nplemma} equivalent to the Nambu-Poisson fundamental
identity. On the other hand, note that for $p>1$, $F'$ is not necessarily a
$(p+1)$-form.

Next, see that the scalar function in front of $\Pi$ in (\ref{eq_thetaprimetrace}) is related to the determinant of the vector bundle isomorphism $1 - \Pi F^{T}$. For $p>1$, any Nambu-Poisson tensor and any $(p+1)$-form $F$, its holds

\begin{equation} \label{detFprimeasfunction}
\det{(1 - \Pi F^{T})} = (1 - \frac{1}{p+1} \<\Pi,F\>)^{p+1}.
\end{equation}
To prove this identity, note  that both sides are scalar functions. We may therefore use any local coordinates on $M$. Again, use those in which $(\ref{eq_Pidecomposition})$ holds. The rest of the proof is straightforward.

Further on, assume that $F$ is closed, that is at least locally $F =
dA$ for a $p$-form $A$. Define a $1$-parametric family of tensors
$\Pi'_{t} := (1 - t\Pi F^{T})^{-1} \Pi$, cf. Footnote \ref{invertible}. This is obviously
chosen so that $\Pi'_{0} = \Pi$ and $\Pi'_{1} = \Pi'$.
Differentiation of $\Pi'_{t}$ with respect to $t$ gives:
\begin{equation}
\partial_{t} \Pi'_{t} = \Pi'_{t} F^{T} \Pi'_{t}.
\end{equation}

This equation can be rewritten as
\begin{equation} \label{eq_difffortheta}
\partial_{t} \Pi'_{t} = -\mathcal{L}_{A^{\sharp}_{t}} \Pi'_{t},
\end{equation}
where the time-dependent vector field $A^{\sharp}_{t}$ is defined as
$A^{\sharp}_{t} = {\Pi'}^{\sharp}_{t}(A)$. To see this, note that
$\Pi'_{t}$ is, using similar arguments as above, a Nambu-Poisson
tensor. Then recall the property (\ref{npfieq1}), and choose $\xi =
A$ and $\eta = dy^{J}$. Contracting the resulting vector field
equality with $dy^{i}$ gives exactly $\mathcal{L}_{A^{\sharp}_{t}}
\Pi'_{t} = - \Pi'_{t} F^{T} \Pi_{t}'$. Equation
(\ref{eq_difffortheta}) states precisely that the flow $\phi_{t}$
corresponding to $A_{t}^{\sharp}$, together with condition
$\Pi'_{0} = \Pi$, maps $\Pi_{t}$ to $\Pi$, that is,
\begin{equation}
\phi_{t}^{\ast}(\Pi'_{t}) = \Pi.
\end{equation}
We have thus found the map $\rho_{A} \equiv \phi_{1}$, which gives
$\rho_{A}^{\ast}(\Pi') = \Pi$. This is the $p \geq 1$ analogue
of the well known semiclassical Seiberg-Witten map. Obviously, it preserves the singular foliation defined by $\Pi$.  We emphasize the
dependence of this map on the $p$-form $A$ by an explicit addition
of the subscript $A$.

Denote ${J^{i}}_{k} = \frac{\partial \widehat X^{i}}{\partial x^{k}}$, with $\widehat X^{i}:=\rho_A^{\ast}(x^i)$ being {\it covariant} coordinates. We
have
\begin{equation} \label{JJJ} \rho_A^{\ast}( {\Pi'}^{j_1,\ldots j_{p+1}} ) = J^{j_1}_{i_1}\ldots J^{j_{p+1}}_{i_{p+1}} \Pi^{i_1\ldots i_{p+1}}.
\end{equation}

Further, denote by $|J|$ the determinant of $J^{i}_{k}$ in some (arbitrarily) chosen local coordinates $(x^{1}, \dots, x^{n})$ on $M$. One can choose, for instance, the special coordinates $(\~x^{i}, \dots \~x^{n})$ on $M$ in which (\ref{eq_Pidecomposition}) holds. We will use the notation $|\~J|$ for the determinant of the matrix $\~J^{i}_{k} = \frac{\partial \~x^{i}(\rho_{A}(x))}{\partial \~x^{k}}$. From now, for any function $\varphi$ (e.g., a matrix component, determinant, etc.), the symbol $\widehat{\varphi}$ will always denote the function defined as $\widehat{\varphi}(x) \equiv \rho_{A}^{\ast}(\varphi)(x) = \varphi(\rho_{A}(x))$. Recall now the definition (\ref{def_pidensity}) of the density $|\Pi(x)|$\footnote{For $p=1$, one can (around every regular point of the characteristic distribution) define $|\Pi(x)|$ to be the Jacobian of the transformation to the Darboux-Weinstein coordinates. This gives a good definition even if $\Pi$ is degenerate.}. By definition of $|J|$, we then have
\begin{equation} \label{eq_JastildJ}
 |J| = |\~J| \frac{|\widehat{\Pi}(x)|^{\frac{1}{p+1}}}{|\Pi(x)|^{\frac{1}{p+1}}}
\end{equation}
The Jacobian $|\~J|$ can easily  be calculated using (\ref{eq_thetaprimetrace}) and (\ref{JJJ}). Indeed, the equation (\ref{JJJ}) can be, in $(\~x)$ coordinates, rewritten as
\[ (1 - \frac{1}{p+1} \<\widehat{\Pi},\widehat{F}\>)^{-1} \epsilon^{j_{1} \dots j_{p+1}} = \epsilon^{j_{1} \dots j_{p+1}} \~J_{i_{1}}^{1} \dots \~J_{i_{p+1}}^{p+1} \epsilon^{i_{1} \dots i_{p+1}}. \]
To justify this, note that Seiberg-Witten map acts nontrivially only in the directions of the first $(p+1)$-coordinates. The Jacobi matrix $\~J$ of $\rho_{A}$ in $(\~x)$ coordinates is thus a block upper triangular with identity matrix in the bottom right block. Moreover, the determinant of $\~J$ is then equal to the determinant of the top left block. We can divide both sides with $\epsilon^{j_{1} \dots j_{p+1}}$. We thus remain with the equation
\[ (1 - \frac{1}{p+1} \<\widehat{\Pi},\widehat{F}\>)^{-1} = \~J_{i_{1}}^{1} \dots \~J_{i_{p+1}}^{p+1} \epsilon^{i_{1} \dots i_{p+1}}= |\~J|. \]
Putting this back into (\ref{eq_JastildJ}), we obtain the useful relation
\begin{equation}
|J|^{p+1} = (1 - \frac{1}{p+1} \< \widehat{\Pi}, \widehat{F} \>)^{-(p+1)} \frac{|\widehat{\Pi}(x)|}{|\Pi(x)|},
\end{equation}
or using (\ref{detFprimeasfunction})\footnote{For $p=1$, one can derive this relation by calculating $|\~J|$ in Darboux-Weinstein coordinates directly from (\ref{JJJ}) and the definition of $\Pi'$, and then use (\ref{eq_JastildJ}).}
\begin{equation} \label{Jacobian}
|J|^{p+1} = \det{(1 - \widehat{\Pi} \widehat{F}^{T})}^{-1} \frac{|\widehat{\Pi}(x)|}{|\Pi(x)|}.
\end{equation}

Note that this expression does not depend on the choice of the Darboux coordinates in which the densities $|\Pi(x)|$ are calculated. We discuss this subtlety in the appendix \ref{sec_NPstructures} under (\ref{eq_Piintwocoordinates}). We see that $|\Pi(x)|$  itself transforms as in (\ref{eq_Pidensitytransform}). Fortunately, the determinant of the block $M$ in (\ref{eq_Piintwocoordinates}) does not depend on the coordinates $(\~x^{1}, \dots \~x^{p+1})$. Since these are the only coordinates changed by the Seiberg-Witten map, we get $(\det{M})(x) = (\det{M})(\rho_{A}(x))$. In other words, these determinants cancel out in the fraction $|\widehat{\Pi}(x)| / |\Pi(x)|$, as expected.

The following observation is in order: The Nambu-Poisson tensor $\Pi_t$ does not depend on the
choice of the gauge $p$-potential $A$. As already mentioned, the Nambu-Poisson map $\rho_{A}$ does: An infinitesimal gauge transformation $\delta A = d \lambda$ -- with a
$(p-1)$-form gauge transformation parameter $\lambda$ -- induces a change in the flow, which is generated by the vector field $X_{[\lambda, A]} = \Pi^{iJ} d\Lambda_J\partial_i$, where
\begin{equation}
\Lambda= \sum_{k=0}^{\infty} \frac{(\mathcal{L}_ {A_{t}^\sharp} +
\partial_t)^k(\lambda)}{(k+1)!}\Big| _{t=0}  \,,
\end{equation}
is the semiclassically noncommutative $(p-1)$-form gauge parameter.
This is the $p$-brane analog of the exact Seiberg-Witten map for the gauge transformation parameter. It is straightforwardly obtained by application of the BCH formula to $\rho^\ast_{A+d\lambda}(\rho^\ast_A)^{-1}$.
Finally, in analogy with the $p=1$ case, we define the (components of the) semiclassically noncommutative field strength to be
\begin{equation}\label{NCF}
\widehat F'_{i_1,\ldots, i_{p+1}} =\rho^{\ast}_A F'_{i_1,\ldots, i_{p+1}},
\end{equation}
i.e., the components of $F'$ evaluated in the covariant coordinates. Infinitesimally, components of $\hat F$ transform as
\begin{equation}\label{NCF2}
\delta\widehat F' =\Pi^{iJ} d\Lambda_J\partial_i \hat F',
\end{equation}
which justifies the adjectives ``semiclassically noncommutative".

\section{Nambu gauge theory; Equivalence of commutative and semiclassically
noncommutative DBI action} \label{DBI}

Here we consider a system of multiple open M2 branes ending on an M5 brane. We would like to describe this system by an
effective action that is exact, for slowly varying fields,  to all orders in the coupling constant. Since we focus only on the bosonic part of this action, we do not need to restrict ourselves to the values $p=2$ and $p'=5$ and our construction is valid for arbitrary values of $p$ and $p'$ such that $p\leq p'$.
Our goal is thus the construction of an effective action for a $p'$-brane with open $p$-branes ending on it while being submerged in a  $C_{p+1}$-background.
The construction is based on two guiding principles: Firstly, this effective action should have dual descriptions similar to the commutative and non-commutative ones of the D-brane and open strings\footnote{Actually, our exposition so far closely followed our previous work \cite{Jurco:2012gc}, where the role of generalized geometry was emphasized.} and secondly, it should feature expressions that also appear in the $p$-brane action (\ref{def_actiontotmatrixF}).

Denote the $p'$-brane submanifold as $N$. We shall now clarify the geometry underlying the following discussion. Originally, $g,\~g,C$ were assumed to be the closed membrane backgrounds in the ambient background manifold $M$. Hereafter, we denote by the same characters their pullbacks to the $p'$-brane $N$. This makes sense since all of them are covariant tensor fields on $M$. Little subtlety comes with the Nambu-Poisson tensor $\Pi$. We have basically two options. First, we would like to restrict some Nambu-Poisson tensor in $M$ to the $p'$-brane. This in fact requires $N$ to be a Nambu-Poisson submanifold of $M$. The latter option is to \emph{choose} the Nambu-Poisson tensor $\Pi$ on $N$ after we restrict the other backgrounds to $N$. The open membrane variables $G,\~G,\Phi$, calculated using the membrane open-closed relations (\ref{eq_occorrespondence1} - \ref{eq_occorrespondence4}), are assumed to be calculated entirely on $N$, using the pullbacks of closed variables. Finally, the field $F$ is assumed to be a $(p+1)$-form defined and having components only in $N$. All the discussion related to Seiberg-Witten map in the previous section is assumed to take place on the submanifold $N$.

The open-closed membrane relations (\ref{openclosedF}) immediately imply
\begin{multline} \label{miracle}
\det[g + (C+F) \tilde g^{-1} (C+F)^T] =  \det{}^2[1-F \Pi^{T}]\cdot \det[G+(\Phi + F') \tilde G^{-1} (\Phi + F')^T] \,,
\end{multline}
where $F' = (I - F \Pi^T)^{-1} F$. Obviously, in order get a sensible action we have to form an integral density, which can be integrated over the world volume of the larger $p'$-brane. And, in order to obtain a noncommutative action from the right hand side of \eqref{miracle}, we have to apply  the Seiberg-Witten map $\rho^{\ast}_A$ to it. It would be tempting to take the square root of the identity \eqref{miracle} to construct the action. But, recall \eqref{Jacobian} and notice the factor $\det{}^{-(p+1)}[1-F\Pi^T]$ appearing in it upon the application of the Seiberg-Witten map. Hence, not the square root but the $2(p+1)$-th root of \eqref{miracle} is the most natural choice to enter the effective action that we look for. As we already said, the Lagrangian density  must be an integral density, and therefore we need to multiply that piece of the action by a proper power of the determinant of the pullback of the target space metric. These considerations fix the action essentially uniquely and we postulate \vspace{-1ex}
\begin{equation} \label{pDBI}
S_\text{$p$-DBI} = -\int d^{p'+1} x \, \frac{1}{g_m} \det{}^\frac{p}{2(p+1)}(g) \cdot \det{}^{\frac{1}{2(p+1)}}\big[g + (C+F) \tilde g^{-1} (C+F)^T \big] \,,
\end{equation}
where $g_m$ is a ``closed membrane'' coupling constant. The integration is over the $p'$-brane and the fields $g$, $\tilde g$, and $C$ in this expression are the pull-backs of the corresponding background target space fields to this $p'$-brane.
Asking for
\begin{multline} \label{MI}
\frac{1}{g_m} \det{}^\frac{p}{2(p+1)}g \cdot \det{}^{\frac{1}{2(p+1)}}\big[g + (C+F) \tilde g^{-1} (C+F)^T \big] \\
= \frac{1}{G_m} \det{}^\frac{p}{2(p+1)}(G)\det{}^{\frac{1}{(p+1)}}[1-\Pi F^T] \cdot \det{}^{\frac{1}{2(p+1)}}\big[G + (\Phi+F') \tilde G^{-1} (\Phi+F')^T \big],
\end{multline}
it follows from \eqref{miracle} that the closed and open coupling constants $g_m$ and $G_m$ must be related as
\begin{equation}
G_m = g_m \left(\det G/\det g\right)^{\frac{p}{2(p+1)}} \,.
\end{equation}

As desired, the action (\ref{pDBI}) is exactly equal to its ``noncommutative'' dual
\vspace{-1ex}
\begin{equation}\label{pNCDBI}
S_\text{$p$-NCDBI} =  -\int d^{p'+1} x \, \frac{1}{\widehat G_m} \,\frac{\widehat{|\Pi|}^\frac{1}{p+1}}{|\Pi|^\frac{1}{p+1}}  \det{}^{\frac{p}{2(p+1)}} \widehat G
 \cdot \det{}^{\frac{1}{2(p+1)}}\big[ \widehat G + (\widehat \Phi + \widehat F') \widehat{\tilde G}{}^{-1} (\widehat \Phi + \widehat F')^T  \big] \,,
\end{equation}
where as before $\,\widehat{\;}\,$ denotes objects evaluated at covariant coordinates\footnote{Let us emphasize that this is not a coordinate transformation of a tensor. We just evaluate the component functions in different coordinates.} and $\widehat F'$ is the Nambu (NC) field strength \eqref{NCF}.
This follows from integrating of \eqref{MI} followed by the change of integration variables on its right hand side according to the Seiberg-Witten map.

The factor involving the quotient of $\widehat{|\Pi|}$ and $|\Pi|$ vanishes for constant $|\Pi|$, but it is essential for the gauge invariance of (\ref{pNCDBI}) in all other cases.

Let us give two alternative, but equivalent, expressions for the action \eqref{pDBI}, which might turn out to be useful when looking for supersymmetric generalizations. The first one is obvious:
\begin{equation} \label{pDBI1}
S_\text{$p$-DBI} = -\int d^{p'+1} x \, \frac{1}{g_m} \det{}^\frac{1}{2}(g) \cdot \det{}^{\frac{1}{2(p+1)}}\big[1 + g^{-1}(C+F) \tilde g^{-1} (C+F)^T \big] \,.
\end{equation}
A very similar expression can be found using (\ref{eq_detidentity})
\begin{equation} \label{pDBI1b}
S_\text{$p$-DBI} = -\int d^{p'+1} x \, \frac{1}{g_m} \det{}^\frac{1}{2}(g) \cdot \det{}^{\frac{1}{2(p+1)}}\big[1 + \~g^{-1}(C+F)^{T} g^{-1} (C+F) \big] \,.
\end{equation}

For the second one, let us note that $\det \tilde g= \det^{p'\choose p-1}g$, in the case of factorizable $\~g$.
Hence, in this case:
\begin{equation} \label{pDBI2}
S_\text{$p$-DBI} = - \int d^{p'+1}x \frac{1}{g_m}  \det{}^{\frac{p-{p'\choose p-1}}{2(p+1)}}g\cdot \det{}^{\frac{1}{2(p+1)}} \bm{g}{(C+F)}{-(C+F)^{T}}{\~g}.
\end{equation}

Let us note that in the case of a D-brane, i.e., $p=1$, we get indeed the DBI D-brane action. In the other extreme case, $p=p'$, we get\footnote{The notation $S_M$ will be justified later.}

\begin{equation} \label{S_M}
S_M = - \int d^{p+1}x \frac{1}{g_m} \det{}^{\frac{1}{2(p+1)}} \bm{g}{(C+F)}{-(C+F)^{T}}{\~g}.
\end{equation}

Now we can compare our action, e.g, to the DBI part of the M5-brane action in equation (2.9) of \cite{Cederwall:1997gg}, \cite{Bao:2006ef}. Their action is, up to conventions,
\begin{equation}\label{Cederwall}
S'=-\int d^6x\, \sqrt{\det g} \sqrt{1 + \frac{1}{3} { \rm tr}k -  \frac{1}{6} { \rm tr}{k}^2  +
 \frac{1}{18} ({ \rm tr}\,{k})^2}\, ,
\end{equation}
where $k^i_j=(dA+C)^{ikl}(dA +C)_{jkl}$ is the modified field strength.
(See also \cite{Bergshoeff:1996ev}, for an early proposal with a similar index structure.)
The form of the polynomial in $k$ in the action has been determined by lengthy computation based on $\kappa$-symmetry and the requirement of
non-linear self-duality, the self-duality relations being consistently decoupled from the background. More precisely, in \cite{Cederwall:1997gg}, \cite{Bao:2006ef}, it is shown that consistency of the non-linear self-duality is restrictive enough that demanding $\kappa$-symmetry gives its explicit form, which can be obtained without a priori specifying the form of the polynomial in the action. At the same time the projector specifying the $\kappa$-symmetry and the form of the polynomial are determined.  

To our surprise, we found that this action $S'$ can be interpreted as a
low-energy (second order in $k$) approximation of our $p$-DBI action (\ref{pDBI}). Indeed,for $p=2$ and $p'=5$ we have $d^{p'+1}x = d^6 x$, $\frac{1}{2(p+1)} = \frac{1}{6}$ and
\[
det^\frac{1}{6}(1 + k) = \sqrt{1 + \frac{1}{3} { \rm tr}k -  \frac{1}{6} { \rm tr}k^2  +
 \frac{1}{18} ({ \rm tr}\,k)^2 + \ldots}\,.
\]
The fact that two very different approaches (one based on non-linear self-duality and $\kappa$-symmetry, the other on commutative/non-commutative duality) give rise to the same action in the low energy limit is very encouraging and seems to indicate that our proposal can indeed be extended to a full supersymmetric action.

Finally, let us mention that noncommutative structures in the context of the M5 brane have previously been discussed, for example, in \cite{Berman:2001rka} and \cite{Berman:2004jv}. However, the type of noncommutativity discussed in these earlier papers is the well-known deformation of the commutative point-wise multiplication along a (constant) Poisson tensor that already appeared in the $p=1$ string theory case. This is very different from the notion of noncommutativity that we argue to be pertinent for $p>1$ and in particular for the  $p=2$ case relevant for the M5 brane: For $p>1$, we do not deform the commutative product -- our ``noncommutativity'' has rather to be understood in the Nambu-Poisson sense as explained in detail above, cf. the remark at the end of the previous section.

\section{Background independent gauge} \label{sec_BIG}

For $p=1$, assuming that the pullback of the background $2$-form $C$ to the $p'$-brane $N$ is non-degenerate and closed (that is symplectic), one can choose the bivector $\Pi$ to be the inverse of $C$ (that is a Poisson bivector corresponding to the symplectic structure $C$).
Solving the open-closed relations then gives
\begin{equation} \label{eq_classicBIG} G = -Cg^{-1}C , \  \Phi = -C. \end{equation}
This is known as the background independent gauge \cite{Seiberg:1999vs}. Our aim is to generalize this construction for $p \geq 1$, even giving milder assumptions on $C$ for $p=1$.

Let us start on the level of linear algebra first. Assume that $V$ is a finite-dimensional vector space. Let $g$ be an inner product on $V$, and $C \in \Lambda^{2}V^{\ast}$ a $2$-form. Let $P: V \rightarrow V$ denote a projector orthogonal with respect to $g$, such that
\[ \ker(C) = \ker(P), \]
where $C$ is viewed as a map $C: V \rightarrow V^{\ast}$. Then there exists a unique bivector $\Pi \in \Lambda^{2}V$, satisfying
\begin{equation} \label{def_thetapinv}
 \Pi C = P \ , \ P \Pi = \Pi.
\end{equation}
The reader can find the proof of this statement in proposition \ref{tvrz_2formpseudoinverse} of appendix \ref{sec_BIGsupplement}.

Recall that open-closed relations for $p=1$ have the form
\begin{equation}
\frac{1}{g+C} = \frac{1}{G + \Phi} + \Pi.
\end{equation}
This equality can be rewritten as
\begin{equation}
G + \Phi = (1 - (g+C)\Pi)^{-1}(g+C).
\end{equation}
Using (\ref{def_thetapinv}), one gets
\[ G + \Phi = P'^{T}gP' - Cg^{-1}C - C, \]
where $P' = 1-P$. From this we can read of the symmetric and skew-symmetric part to get
\begin{equation} \label{eq_BIGocsolution} G = P'^{T}gP' - Cg^{-1}C \ , \ \Phi = -C. \end{equation}
We can view this as a generalization of (\ref{eq_classicBIG}), not assuming a non-degenerate $C$. See that $G$ is again a positive definite metric, and $G + \Phi$ is thus invertible. Note that we are now on the level of a single vector space $V$, not discussing any global properties of $\Pi$ yet.

We would like to generalize this procedure to $p \geq 1$ case. Our goal is to find a suitable choice for $\Pi$, such that $\Phi = -C$. Assume that $C: \Lambda^{p}V \rightarrow V^{\ast}$ is a linear map, $g$ is an inner product on $V$, and $\~g$ is an inner product on $\Lambda^{p}V$. The key is to keep in mind the open-closed relations (\ref{eq_ocalaSW}). We see that by defining
\[ \G = \bm{g}{0}{0}{\~g} \ , \ \B = \bm{0}{C}{-C^{T}}{0}, \]
we get an inner product $\G$ on $W \equiv V \oplus \Lambda^{p}V$, and a bilinear skew-symmetric form $\B \in \Lambda^{2} W^{\ast}$.

The situation is thus analogous to the previous one, if we replace $V$ by $W$, the metric $g$ by $\G$, and the $2$-form $C$ by $\B$.  If we define $\mathcal{P}$ to be an orthogonal projector with respect to $\G$ with $\ker(\mathcal{P}) = \ker(\B)$, we may again apply  proposition \ref{tvrz_2formpseudoinverse} to see that there exists a unique $\Theta \in \Lambda^{2}W$, such that
\begin{equation} \label{eq_Pg1BIGPinv}
 \Theta \B = \P \ , \P \Theta = \Theta.
\end{equation}
Now we can solve the open-closed relations (\ref{eq_ocalaSW}) for this choice of $\Theta$, using the same calculation as we did in order to obtain (\ref{eq_BIGocsolution}). One gets
\begin{equation} \label{eq_Pg1BIGsolutions}
\H = \P'^{T} \G \P' - \B \G^{-1} \B \ , \ \Xi = -\B,
\end{equation}
where $\P' = 1 - \P$. Exploring what $\B$ and $\Xi$ are, leads to $\Phi = -C$, as intended. However, we do not know whether $\H$ and $\Theta$ obtained by this procedure are of the suitable form, that is whether $\H$ is block-diagonal and $\Theta$ block-off-diagonal. This can be easily proved by examining the projector $\P$. Clearly, one has
\[ \ker{\B} = \ker{C^{T}} \oplus \ker{C} \subseteq V \oplus \Lambda^{p}V. \]
Therefore we have that $\Img(\P) = \ker{\B}^{\perp} = (\ker{C}^{T})^{\perp(g)} \oplus (\ker{C})^{\perp(\~g)}$. This proves that in a block form, we have
\[ \mathcal{P} = \bm{P}{0}{0}{\~P}, \]
where $P: V \rightarrow V$ is an orthogonal projector with respect to $g$, and $\~P: \Lambda^{p}V \rightarrow \Lambda^{p}V$ is an orthogonal projector with respect to $\~g$. This and the relation (\ref{eq_Pg1BIGsolutions}) imply that $\H$ is block-diagonal.
The second equality in (\ref{eq_Pg1BIGPinv}) then proves that $\Theta$ is block-off-diagonal, that is
\[ \Theta = \bm{0}{\Pi}{-\Pi^{T}}{0}, \]
where $\Pi: \Lambda^{p}V^{\ast} \rightarrow V$.
We can now simply extract all the relations from (\ref{eq_Pg1BIGPinv}). The equality $\Theta \B = \P$ gives
\[ \bm{0}{\Pi}{-\Pi^{T}}{0} \bm{0}{C}{-C^{T}}{0} = \bm{P}{0}{0}{\~P}, \]
which translates into
\begin{equation} \label{eq_Cpinv}
\Pi C^{T} = -P \ , \ \Pi^{T}C = -\~P.
\end{equation}
Rewriting the equation $\B \P = \B$, we get
\[ \bm{0}{C}{-C^{T}}{0} \bm{P}{0}{0}{\~P} = \bm{0}{C}{-C^{T}}{0}, \]
which translates into
\begin{equation}
C \~P = C \ , \ C^{T}P = C^{T}.
\end{equation}
Also see that $\ker(\~P) = \ker(C)$, and $\ker(P) = \ker(C^{T})$. The equality $\P \Theta = \Theta$ gives
\[  \bm{P}{0}{0}{\~P} \bm{0}{\Pi}{-\Pi^{T}}{0} = \bm{0}{\Pi}{-\Pi^{T}}{0}, \]
and thus
\begin{equation} \label{eq_ThPcoop}
 P \Pi = \Pi \ , \ \~P \Pi^{T} = \Pi^{T}.
\end{equation}
Finally, we may examine (\ref{eq_Pg1BIGsolutions}) to find
\begin{equation} \label{eq_OCsoltionsBIG}
 G = P'^{T}gP' + C \~g^{-1} C^{T} \ , \ \~G = \~P'^{T}\~g\~P' + C^{T}g^{-1}C \ , \ \Phi = -C.
\end{equation}
We have thus shown that, corresponding to the orthogonal projectors $P$ and $\~P$ and the linear map $C: \Lambda^{p}V \rightarrow V^{\ast}$, there exists a unique linear map $\Pi: \Lambda^{p}V^{\ast} \rightarrow V$, such that (\ref{eq_Cpinv}) and (\ref{eq_ThPcoop}) hold. Plugging this $\Pi$ into open-closed relations (\ref{eq_ocalaSW}) gives (\ref{eq_OCsoltionsBIG}).

To use this for our purposes, we have to impose conditions on $C$ to ensure that $\Pi$ is a Nambu-Poisson tensor.

For $p>1$, first observe that the linear map $\Pi: \Lambda^{p}V^{\ast} \rightarrow V$ induced (at a chosen point on $M$) by a Nambu-Poisson tensor  has rank either $0$ or $p+1$. Since $\Pi$  always has the same rank as $C$, we get the first assumption on the linear map $C$.

There will always arise problems with the smoothness of $\Pi$ at points $x \in N$, where $C(x) = 0$. If this set has measure zero, we can change the area of integration in DBI action from $N$ to an open submanifold $N'$, where $C(x) \neq 0$. If not, we cannot go to the background-independent gauge. Let us hereafter assume that $C(x) \neq 0$ for all $x \in N$, and therefore that $\mbox{rank}(C) = p+1$.

Now assume that the linear map $C$ is induced by a $(p+1)$-form $C \in \Lambda^{p+1}V^{\ast}$. Note that in this case, we always have the estimate $\mbox{rank}(C) \geq p+1$.

Let $D \subseteq V$ denote the non-degenerate subspace of $C^{T}$ orthogonal (with respect to $g$) to its kernel, that is $D = \ker(C^{T})^{\perp}$. Assumption on the rank of $C$ thus means that $\dim(D) = p+1$. From the skew-symmetry of $C$, we have that $C \in \Lambda^{p+1}D^{\ast}$. It is thus a top-level form on $D$. Choose now an orthonormal basis $(e_{1}, \dots, e_{p+1})$ of $D$. We see that
\begin{equation} \label{eq_CinONframe}
 C = \lambda \cdot e^{1} \^ \dots \^ e^{p+1},
\end{equation}
where $\lambda \neq 0$. Now, choosing an arbitrary complementary basis $(f_{1}, \dots, f_{p'-p})$ of $\ker(C^{T}) \equiv D^{\perp}$, one can find counterexamples to the assumption that, for a general $\~g$, the map $\Pi$ is a $(p+1)$-vector (although it has a correct rank). We thus have to add the second assumption: $\~g$ has to be of the special skew-symmetrized tensor product form (\ref{def_tensorofg}).

In this case we find that $\Lambda^{p}D$ is spanned by orthonormal basis of the form $e_{1} \^ \dots \^ \hat{e}_{r} \^ \dots \^ e_{p+1}$. This allows us to write $\Pi$ explicitly as
\begin{equation} \label{eq_Thetacalculated}
\Pi = - \frac{1}{\lambda} \cdot e_{1} \^ \dots \^ e_{p+1}.
\end{equation}
It is easy to show that such a $\Pi$ indeed satisfies (\ref{eq_Cpinv}) and (\ref{eq_ThPcoop}), and since such a $\Pi$ is unique, this is the one. We can thus conclude that for $\mbox{rank}(C) = p+1$, and $\~g$ in the form (\ref{def_tensorofg}), $\Pi$ is a $(p+1)$-vector, more precisely $\Pi \in \Lambda^{p+1} D$.

We now turn our attention to global properties. If we assume that $C(x) \neq 0$ on the $p'$-brane, we can define the subspace $D$ at every point, defining a smooth subbundle (it is an orthogonal complement to the kernel of constant rank vector bundle morphism $C^{T}$). Around any point, we can choose a local orthonormal frame $(e_{1}, \dots, e_{p+1})$, forming a local basis for the sections of $D$. The expression (\ref{eq_Thetacalculated}) proves that $\Pi$ is a smooth $(p+1)$-vector on the $p'$-brane, since $\frac{1}{\lambda}$ is a smooth function.

Finally, we have to decide under which conditions $\Pi$ forms a Nambu-Poisson tensor. In the view of lemma \ref{lem_nptoplevel}, we see that the sufficient and necessary condition is that the subbundle $D$ defines an integrable distribution in $N$. This distribution has to be regular, and thus, this condition is  equivalent to the involutivity of $D$ under vector field commutator: $[D,D] \subseteq D$.

One can find a simple equivalent criterion for $C$ to define an integrable distribution $D$.  In order to do so, assume now that $(e_{1}, \dots, e_{p+1}, f_{1}, \dots f_{p'-p})$ is a positively oriented orthonormal local frame for $N$, such that $(e_{1}, \dots, e_{p+1})$ is a local orthonormal frame for $D$. The metric volume form $\Omega_{g}$ is then by definition
\[ \Omega_{g} = e^{1} \^ \dots \^ e^{p+1} \^ f^{1} \^ \dots \^ f^{p'-p}. \]
Having a volume form, one can form the Hodge dual of $C$. Using (\ref{eq_CinONframe}) we get
\[ \ast C = \lambda \cdot f^{1} \^ \dots \^ f^{p'-p}. \]
We see that $D = \ker(\ast C)^{T}$, $(\ast C)^{T}: TN \rightarrow \Lambda^{p'-p-1} T^{\ast}N$. But forms with integrable kernel distribution have their own name, they are called integrable forms, see Appendix \ref{sec_BIGsupplement} for the definition and basic properties. We can conclude that $\Pi$ is a Nambu-Poisson $(p+1)$-vector if and only if $\ast C$ is an integrable everywhere non-vanishing $(p' - p)$-form on $N$. Note that the Hodge star is defined with respect to the induced metric on $N$.

There exists a nice sufficient integrability condition: If $C$ is a $(p+1)$-form of rank $p+1$, such that $\delta C = 0$, then $\ast C$ is integrable. By $\delta$ we denote the codifferential defined using the Hodge duality. Note that $\delta C = 0$ are the non-homogeneous charge free Maxwell equations for the field strength $C$. Also, note that in the whole discussion, we do not need the integrability of the distribution $D^{\perp}$. Since $C$ is already a non-vanishing $(p+1)$-form of rank $p+1$, the sufficient condition for integrability of $D^{\perp}$ is $dC = 0$. Interestingly, both $D$ and $D^{\perp}$ are integrable regular distributions if $C$ is a $(p+1)$-form of rank $p+1$, satisfying the Maxwell equations $dC = 0$, $\delta C = 0$.

For $p=1$, the discussion is very similar, except that the rank of $C$ can be any nonzero even integer not exceeding $n$. This adds another condition on $dC$. In particular, the necessary and sufficient condition on $C$ to define a Poisson tensor $\Pi$ is the integrability of the  regular smooth distribution $D$, and a condition $dC|_{\Gamma(D)} = 0$.
\section{Non-commutative directions, double scaling limit} \label{sec_NCdirections}
By the construction of the preceding section, we have the decompositions
\[ TM = D \oplus D^{\perp}, \ \TM{p} = \~D \oplus \~D^{\perp}, \]
where $\~D = \Lambda^{p}D$. We say that tangent vectors contained in $D$ point in ``non-commutative" directions. Because $D$ is integrable, around each point there are coordinates  such that $D$ is spanned by coordinate tangent vectors corresponding to first $p+1$ of these coordinates. These local coordinates are accordingly called ``non-commutative" coordinates. This terminology comes from the fact that for $p=1$, we have $\{x^{i},x^{j}\} = \Pi^{ij}$. The right-hand side is non-vanishing when both $x^{i}$ and $x^{j}$ correspond to $D$. This gives non-vanishing quantum-mechanical commutator of these coordinates.

We can thus write all involved quantities in the block matrix form corresponding to this decomposition. From the orthogonality of respective subspaces, the matrices of $g$ and $\~g$ will be block diagonal:
\[ g = \bm{g_{\nc}}{0}{0}{g_{\c}}, \ \~g = \bm{\~\g_{\nc}}{0}{0}{\~g_{\c}}, \]
where $g_{\nc}$ is a positive definite fibrewise metric on $D$, $\g_{\c}$ is a positive definite fibrewise metric on $D^{\perp}$ and $\~g_{\nc}$ and $\~g_{\c}$ are positive definite fibrewise metrics on $\~D$ and $\~D^{\perp}$, respectively. In the same fashion we obtain
\[ C = \bm{C_{\nc}}{0}{0}{0}, \ \Pi = \bm{\Pi_{\nc}}{0}{0}{0} \ , \  F = \bm{F_{\nc}}{F_{\Fi}}{F_{\Fii}}{F_{\c}}. \]
Examine how the $F$-gauged tensor $\Pi'$ looks like in this block form.
We have
\[ 1 - F^{T} \Pi = \bm{1 - F_{\nc}^{T} \Pi_{\nc}}{0}{-F_{\Fi}^{T} \Pi_{\nc}}{1}. \]
Hence
\[ \Pi' \equiv \Pi (1 - F^{T}\Pi)^{-1} = \bm{ \Pi_{\nc} (1 - F_{\nc}^{T} \Pi_{\nc})^{-1}}{0}{0}{0}. \]
Denote $\Pi_{\nc}' = \Pi_{\nc}(1 - F_{\nc}^{T} \Pi_{\nc})^{-1}$. We also have $\Pi_{\nc}' = (1 - \Pi_{\nc} F_{\nc}^{T})^{-1} \Pi_{\nc}$. Also, note that in this formalism $P$ and $\~P$ are simply given as
\[ P = \bm{1}{0}{0}{0} \ , \ \~P = \bm{1}{0}{0}{0}. \]
Hence, the defining equations of $\Pi$ can be written as
\begin{equation} \label{eq_tildCinv}
 \Pi_{\nc} C_{\nc}^{T} = -1 \ , \ \Pi_{\nc}^{T} C_{\nc} = -1.
\end{equation}

Having this in hand, recall that for $p=1$, the background independent gauge could be obtained in a completely different way. It was obtained by Seiberg and Witten in \cite{Seiberg:1999vs} as a following limit of the relation (\ref{eq_p1correspondence}). Reintroducing the Regge slope $\alpha'$ into description, the relation between closed variables $g$, $C$ and Nambu fields $G_{\NB}$, $\Pi_{\NB}$ is explicitly
\[ G_{\NB} = g - (2\pi \alpha')^{2} Cg^{-1}C^{T}, \ \frac{1}{2 \pi \alpha'} \Pi_{\NB} = -(2 \pi \alpha')g^{-1}C \big(g - (2\pi \alpha')^{2}Cg^{-1}C\big)^{-1}. \]
Now one would like to do the zero slope limit $\alpha' \rightarrow 0$ in a way such that $G_{\NB}$ and $\Pi_{\NB}$ remain finite. This clearly requires the simultaneous scaling of the metric $g$. Scaling the $g$ as a whole will not work, since the resulting $G_{\NB}$ will not be a metric. The correct answer is given by scaling the non-commutative part $g_{\nc}$ and commutative part $g_{\c}$ of the metric $g$ differently. The resulting maps $G_{\NB}$ and $\Pi_{\NB}$ also split accordingly as
\[ G_{\NB \nc} = g_{\nc} - (2\pi \alpha')^{2} C_{\nc} g_{\nc}^{-1} C_{\nc}^{T}, \ G_{\NB \c} = g_{\c}, \]
\[ \frac{1}{2 \pi \alpha'}\Pi_{N \nc} =  -(2\pi \alpha') g_{\nc}^{-1} C_{\nc} (g_{\nc} - (2 \pi \alpha')^{2}C_{\nc} g_{\nc}^{-1}C_{\nc})^{-1}. \]
Now, scaling $g_{\nc} \propto \epsilon$, $g_{\c} \propto 1$, $\alpha' \propto \epsilon^{\frac{1}{2}}$ as $\epsilon \mapsto 0$ gives in this limit
\[ G_{N \nc} = - C_{\nc} g_{\nc}^{-1} C_{\nc}^{T}, \ G_{N \c} = g_{\c}, \]
\[ \Pi_{N \nc} = C_{\nc}^{-1}. \]
Replacing $\Pi_{N}$ by $\Pi$ and $G_{N}$ by $G$ is exactly the background independent gauge. This double scaling limit was then used to determine which terms should be kept in the expansion of the DBI action. We would like to find an analogue of this in our $p>1$ case.\footnote{See \cite{Chen:2010br} for a previous discussion of the double scaling limit in the context of the M2/M5 system that came to different conclusions regarding the appropriate powers of $\epsilon$.} We immediately see that first naive answer would be wrong. One of the relations is
\[ G_{\NB \nc} = g_{\nc} + C_{\nc} \~g_{\nc}^{-1} C_{\nc}^{T}. \]
Note that $\~g_{\nc}$ is again a skew-symmetrized $p$-fold tensor product of $g_{\nc}$. This suggests that if $g_{\nc} \propto \epsilon$, then $\~g_{\nc} \propto \epsilon^{p}$. This would imply that $C_{\nc} \propto \epsilon^{\frac{p}{2}}$ in order to keep $G_{N \nc}$ finite (we have included $\epsilon$ into $C$). But the second relation is
\[ \~G_{N \bullet} = \~g_{\nc} + C_{\bullet}^{T} g_{\bullet}^{-1} C. \]
This shows that $\~G_{N} \rightarrow 0$ as $\epsilon \rightarrow 0$. This is clearly not very plausible. However, this can still be fixed by using the remaining gauge fixing freedom of the Polyakov action (\ref{def_polyakov}) by scaling also the ratio between $g$ and $\~g$. The biggest issue comes with the fact that $\~g_{\c}$ is not a tensor product of $g_{\c}$'s only. In fact, every component $(\~g_{\c})_{IJ}$ contains as many $g_{\nc}$'s as the number of ``commutative" indices in $I$ (or $J$) is. This means that every component of $\~g_{\c}$ should scale differently. We must thus abandon the idea of scaling just $g$, we have to scale $\~g$ independently! The correct answer is given by the geometry of the vector bundle $W= TM \oplus \Lambda^p TM$ again. We immediately see that scaling $\G_{\nc} \propto \epsilon$, $\G_{\c} \propto 1$ and $\B \propto \epsilon^{\frac{1}{2}}$ gives in limit $\epsilon \rightarrow 0$ the background independent gauge. This corresponds to
\begin{equation} \label{def_doublescaling}
 g_{\nc} \propto \epsilon, \ \~g_{\nc} \propto \epsilon, \ g_{\c} \propto 1, \ \~g_{\c} \propto 1, \ C_{\nc} \propto \epsilon^{\frac{1}{2}}.
\end{equation}
Let us note that in the case of an M5 brane a scaling treating directions differently was described in \cite{Bergshoeff:2000ai} and \cite{Bergshoeff:2000jn}. It would be interesting to compare the scaling in these papers with the one introduced here. 
\section{Matrix model} \label{sec_matrix}

Now we will apply the previous generalization of the background independent gauge.
We will use the double scaling limit to cut off the power series expansion of the DBI action. It turns out that we find an action describing a natural $p>1$ (semi-classical) analogue of a matrix model with higher brackets and an interacting with the gauge field $F$. It will be of order $2(p+1)$ in the matrix variables $\widehat{X}^a$,  and at most quadratic in  $F$. The term of order $2(p+1)$ in $\widehat{X}^a$'s and constant in $F$ gives a possible $p>1$ analogue of the semiclassical pure matrix model.

Assume that $C$ satisfies all the conditions required for $\Pi$ to be a Nambu-Poisson tensor on $N$.
From (\ref{pDBI1}), we have that Lagrangian of the commutative $p$-DBI action has the form
\[ \mathcal{L}_{p-\text{DBI}} = - \frac{1}{g_{m}} \det{}^{\frac{1}{2}}(g) \cdot \det{}^{\frac{1}{2(p+1)}}[ 1 + g^{-1}(C+F)\~g^{-1}(C+F)^{T}]. \]
Note that the second determinant is the determinant of the vector bundle endomorphism $X: TM \rightarrow TM$, where $X = 1 + g^{-1}(C+F)\~g^{-1}(C+F)^{T}$.
In the block form $X: D \oplus D^{\perp} \rightarrow D \oplus D^{\perp}$, we have
\[ X = \bm{1 + g_{\nc}^{-1}(C_{\nc}+F_{\nc})\~g_{\nc}^{-1}(C_{\nc}+F_{\nc})^{T} + g_{\nc}^{-1}F_{\Fi}\~g_{\c}^{-1}F_{\Fi}^{T}}{g_{\nc}^{-1}(C_{\nc}+F_{\nc})\~g_{\nc}^{-1}F_{\Fii}^{T} + g_{\nc}^{-1}F_{\Fi}\~g_{\c}^{-1}F_{\c}^{T}}{g_{\c}^{-1}F_{\Fii}\~g_{\nc}^{-1}(C_{\nc}+F_{\nc})^{T} + g_{\c}^{-1}F_{\c}\~g_{\c}^{-1}F_{\Fi}^{T}}{1 + g_{\c}^{-1}F_{\Fii}\~g_{\nc}^{-1}F_{\Fii}^{T} + g_{\c}^{-1} F_{\c} \~g_{\c}^{-1} F_{\c}^{T}}. \]
Here we have used the following notations for the blocks of $F$
\[ F = \bm{F_{\nc}}{F_{\Fi}}{F_{\Fii}}{F_{\c}}. \]
This can be decomposed as a product
\[ X = \bm{g_{\nc}^{-1}(C_{\nc} + F_{\nc})}{0}{0}{1} Y \bm{\~g_{\nc}^{-1}(C_{\nc} + F_{\nc})^{T}}{0}{0}{1}, \]
where the vector bundle endomorphism $Y: \~D \oplus D^{\perp} \rightarrow \~D \oplus D^{\perp}$ is
\[ Y = \bm{1 + \Pi_{\nc}'^{T}(g_{\nc} + F_{\Fi} \~g_{\c}^{-1} F_{\Fi}^{T}) \Pi_{\nc}' \~g_{\nc}}{\~g_{\nc}^{-1}(F_{\Fii}^{T} - \~g_{\nc} \Pi_{\nc}'^{T}F_{\Fi} `\~g_{\c}^{-1} F_{\c}^{T})}{g_{\c}^{-1}(F_{\Fii} - F_{\c} \~g_{\c}^{-1}F_{\Fi}^{T} \Pi_{\nc}' \~g_{\nc})}{1 + g_{\c}^{-1}F_{\Fii}\~g_{\nc}^{-1}F_{\Fii}^{T} + g_{\c}^{-1}F_{\c}\~g_{\c}^{-1}F_{\c}^{T}}. \]
Writing $Y$ in block form as
\[ Y = \bm{Y_{\nc}}{Y_{\Fi}}{Y_{\Fii}}{Y_{\c}}, \]
note that $Y_{\nc}$ is an invertible matrix. This is true because it is a top left block of the matrix $Y$ coming from positive definite matrix $g + (C+F)\~g^{-1}(C+F)$ by multiplying it by invertible block-diagonal matrices. Hence, we can write
\begin{equation} \det{(Y)} = \det{(Y_{\nc})} \det{ (Y_{\c} - Y_{\Fi} Y_{\nc}^{-1} Y_{\Fii})}. \label{detY1}\end{equation}
The second matrix has the form
\[
\begin{split}
Y_{\c} - Y_{\Fi} Y_{\nc}^{-1} Y_{\Fii} &= 1 + g_{\c}^{-1} F_{\Fii} (1 - Y_{\nc}^{-1})\~g_{\nc}^{-1} F_{\Fii}^{T} + g_{\c}^{-1} F_{\c} \~g_{0}^{-1} F_{\c}^{T} + g_{\c}^{-1} F_{\Fii} Y_{\nc}^{-1} \Pi_{\nc}'^{T} F_{\Fi} \~g_{\c}^{-1} F_{\c}^{T} \\
& + g_{\c}^{-1} F_{\c} \~g_{\c}^{-1} F_{\Fi}^{T} \Pi_{\nc}' \~g_{\nc} Y_{\nc}^{-1} \~g_{\nc}^{-1} F_{\Fii}^{T} - g_{\c}^{-1} F_{\c} \~g_{\c}^{-1} F_{\Fi}^{T} \Pi_{\nc}' \~g_{\nc} Y_{\nc}^{-1} \Pi_{\nc}'^{T} F_{\Fi} \~g_{\c}^{-1} F_{\c}^{T}.
\end{split}
\]
At this point, we will employ the double scaling limit introduced above. Namely, in the $\det^{\frac{1}{2(p+1)}}(Y)$, we wish to keep only the terms scaling at most as $\epsilon^{1}$. Note that $(Y_{\nc} - 1) \propto \epsilon$. Also, $Y_{\nc}^{-1} = 1 - (Y_{\nc} - 1) + o(\epsilon^{2})$. Using this, we can write
\[
\begin{split}
Y_{\c} - Y_{\Fi} Y_{\nc}^{-1} Y_{\Fii} = 1 + g_{\c}^{-1} \Big( F_{\Fii} \Pi_{\nc}'^{T} g_{\nc} \Pi_{\nc}' F_{\Fii}^{T} +  \big( F_{\Fii} \Pi_{\nc}'^{T} F_{I} + F_{\c} \big) \~g_{\c}^{-1} \big( F_{\Fii} \Pi_{\nc}'^{T} F_{I} + F_{\c} \big)^{T} \Big) + o(\epsilon^{2}).
\end{split}
\]
The whole term in parentheses after $g_{0}^{-1}$ is of order $\epsilon^{1}$. Therefore, we have
\[
\begin{split}
\det{}^{\frac{1}{2(p+1)}}(Y_{\c} - Y_{\Fi} Y_{\nc}^{-1} Y_{\Fii}) & = 1 + \frac{1}{2(p+1)} \tr ( g_{\c}^{-1} F_{\Fii} \Pi_{\nc}'^{T} g_{\nc} \Pi_{\nc}' F_{\Fii}^{T} ) \\
& +  \frac{1}{2(p+1)} \tr \Big( g_{\c}^{-1} \big( F_{\Fii} \Pi_{\nc}'^{T} F_{I} + F_{\c} \big) \~g_{\c}^{-1} \big( F_{\Fii} \Pi_{\nc}'^{T} F_{I} + F_{\c} \big)^{T} \big) \Big) + o(\epsilon^{2}).
\end{split}
\]
For the first factor in (\ref{detY1}), we have
\[ \det{}^{\frac{1}{2(p+1)}}(Y_{\nc}) = 1 + \frac{1}{2(p+1)} \tr\big( \Pi_{\nc}'^{T} (g_{\nc} + F_{\Fi} \~g_{\c}^{-1} F_{\Fi}^{T}) \Pi_{\nc}' \~g_{\nc} \big) + o(\epsilon^{2}). \]
Putting all together, we obtain
\begin{equation} \label{eq_detYexpansion}
\begin{split}
\det{}^{\frac{1}{2(p+1)}}(Y) & = 1 + \frac{1}{2(p+1)} \tr\big( \Pi_{\nc}'^{T} (g_{\nc} + F_{\Fi} \~g_{\c}^{-1} F_{\Fi}^{T}) \Pi_{\nc}' \~g_{\nc} \big) + \frac{1}{2(p+1)} \tr ( g_{\c}^{-1} F_{\Fii} \Pi_{\nc}'^{T} g_{\nc} \Pi_{\nc}' F_{\Fii}^{T} ) \\
& +  \frac{1}{2(p+1)} \tr \Big( g_{\c}^{-1} \big( F_{\Fii} \Pi_{\nc}'^{T} F_{I} + F_{\c} \big) \~g_{\c}^{-1} \big( F_{\Fii} \Pi_{\nc}'^{T} F_{I} + F_{\c} \big)^{T} \Big) + o(\epsilon^{2}).
\end{split}
\end{equation}
Now, comparing the definitions of scalar densities corresponding to $\Pi$ and  $\Pi'$, it is clear that
\[ \det(C_{\nc} + F_{\nc}) = \pm \det(1 - \Pi F^{T}) \cdot |\Pi(x)|^{-(p+1)}. \]
Here we assume that one chooses the basis of $\Lambda^{p}D$ induced by the basis of $D$. The sign $\pm$ depends on the ordering of that basis. Next, see that $\det(\~g_{\nc}) = \det{}^{\binom{p}{p-1}}(g_{\nc}) = \det{}^{p}(g_{\nc})$. This shows that
\[ S_\text{$p$-DBI} = \mp \int d^{p'+1}x \frac{1}{g_{m}} \frac{\det{}^{\frac{1}{p+1}}(1 - \Pi F^{T})}{|\Pi(x)|^{\frac{1}{p+1}} \det{}^{\frac{1}{2}}(g_{\nc})} \det{}^{\frac{1}{2}}(g) \det{}^{\frac{1}{2(p+1)}}(Y). \]
Changing the coordinates according to Seiberg-Witten map, we get the noncommutative DBI action in the form:
\[ S_\text{$p$-NCDBI} = \mp \int d^{p'+1}x \frac{1}{\widehat{g}_{m}} \frac{\det{}^{\frac{1}{2}}(\widehat{g})}{|\Pi(x)|^{\frac{1}{p+1}} \det{}^{\frac{1}{2}}(\widehat{g}_{\nc})}  \det{}^{\frac{1}{2(p+1)}}(\widehat{Y}). \]

In the last part of the discussion assume that the distribution $D^{\perp}$ is also integrable, so we can use the set of local coordinates $(x^{1}, \dots, x^{p+1}, x^{p+2}, \dots, x^{p'+1})$ on $N$, such that $(\ppx{}{1}, \dots \ppx{}{p+1})$ span $D$, and $(\ppx{}{p+2}, \dots, \ppx{}{p'+1})$ span $D^{\perp}$. All quantities with indices in $D^{\perp}$ are now assumed to be in this coordinate basis. Under this assumptions, the integral density in the action can be written as
\[ \det{}^{\frac{1}{2}}(g) = \det{}^{\frac{1}{2}}(g_{\nc}) \cdot \det{}^{\frac{1}{2}}(g_{\c}). \]

Finally, to distinguish the noncommutative and commutative coordinates, we reserve the letters $(a,b,c)$ for labeling the coordinates $(x^{1}, \dots, x^{p+1})$, $(i,j,k)$ for labeling the coordinates $(x^{p+2},
\dots, x^{p'+1})$, $(A,B,C)$ for $p$-indices containing only noncommutative indices (thus $p$-indices labeling $\~D$) and $(I,J,K)$ for $p$-indices containing at least one commutative index (thus $p$-indices labeling $\~D^{\perp})$.  Also, note that from the definition of $\rho_{A}$, we have
\[ \widehat{\Pi}^{'aB} = \{ \widehat{X}^{a}, \widehat{X}^{b_{1}}, \dots, \widehat{X}^{b_{p}} \}, \]
where $\{\cdot,\dots,\cdot\}$ is the Nambu-Poisson bracket corresponding to $\Pi$, $\widehat{X}^{a} = \rho_{A}^{\ast}(x^{a})$, and $B = (b_{1}, \dots, b_{p})$. To simplify the expressions, we shall also use the shorthand notation $\{\cdot,\widehat{X}^{A}\} \equiv \{\cdot, \widehat{X}^{a_{1}}, \dots, \widehat{X}^{a_{p}} \}$. Finally, we also introduce usual index raising/lowering conventions, for example, $\Fud{k}{A} = \sum_{n=1}^{p'+1} \widehat{g}^{kn} \widehat{F}_{nA} = \widehat{g}^{kl} \widehat{F}_{lA}$, or  $\Fdu{k}{A} = \widehat{\~g}^{AB} \widehat{F}_{kB}$ for multiindices. Note that since both $g$ and $\~g$ are block diagonal, no confusion concerning range of summation appears. Implementing this notation, we can write
\[
\begin{split}
S_\text{$p$-NCDBI} = \mp \int d^{p'+1}x \frac{1}{\widehat{g}_{m}} \frac{\det{}^{\frac{1}{2}}(\widehat{g}_{\c})}{|\Pi(x)|^{\frac{1}{p+1}}} \Big( 1 + \frac{1}{2(p+1)}  \{ \widehat{X}^{a}, \widehat{X}^{A}\} \{ \widehat{X}_{a},\widehat{X}_{A} \} \\
+ \frac{1}{2(p+1)} \{\widehat{X}^{a}, \widehat{X}^{A} \} \Fdu{a}{I} \widehat{F}_{bI} \{\widehat{X}^{b}, \widehat{X}_{A} \} + \frac{1}{2(p+1)} \{\widehat{X}^{a}, \widehat{X}^{A}\} \widehat{F}_{kA} \Fud{k}{B} \{ \widehat{X}_{a}, \widehat{X}^{B} \} \\
+ \frac{1}{2(p+1)}(\widehat{F}_{kA} \{ \widehat{X}^{a}, \widehat{X}^{A} \} \widehat{F}_{aJ} + \widehat{F}_{kJ}) (\Fud{k}{B} \{ \widehat{X}^{b}, \widehat{X}^{B} \} \Fdu{b}{J} + \widehat{F}^{kJ}) \Big) + \cdots.
\end{split}
\]
Note that the first non-cosmological term $\{\widehat{X}^{a},\widehat{X}^{A}\} \{ \widehat{X}_{a}, \widehat{X}_{A}\}$ can be rewritten as
\begin{equation}
 \{ \widehat{X}^{a}, \widehat{X}^{A}\} \{ \widehat{X}_{a},\widehat{X}_{A} \} = \frac{1}{p!} \widehat{g}_{a_{1}b_{1}} \dots \widehat{g}_{a_{p+1} b_{p+1}} \{ \widehat{X}^{a_{1}}, \dots, \widehat{X}^{a_{p+1}} \} \{ \widehat{X}^{b_{1}}, \dots, \widehat{X}^{b_{p+1}} \},
\end{equation}
where summation now goes over all (not strictly ordered) $(p+1)$-indices $(a_{1}, \dots, a_{p+1})$ and $(b_{1}, \dots, b_{p+1})$. Here, we have used the fact that $\~g_{\nc}$ is a skew-symmetrized $p$-fold tensor product of $g_{\nc}$. We can even drop the restriction of the summations to noncommutative directions, since the Nambu-Poisson bracket takes care of this automatically. This term corresponds to a $p>1$ generalization of the matrix model.
Note that using the double scaling limit for the expansion of (\ref{eq_detYexpansion}) leads to a series in positive integer powers of $\epsilon$, automatically truncating higher-order powers in $F$. This gives an independent justification of the independent scaling of $\~g_{\nc}$ and $\~g_{\c}$ in (\ref{def_doublescaling}).

\section{Conclusions and Discussion}
In this paper we have extended, clarified and further developed the construction outlined in \cite{Jurco:2012yv}. We discussed in detail the bosonic part of an all-order effective action for a system of multiple $p$-branes ending on a $p'$-brane. The leading principle was to have an action allowing, similarly to the DBI action, for two mutually equivalent descriptions: a commutative and a ``noncommutative'' one. As explained in the main body of the paper, the noncommutativity means a semicalssical one, in which the Poisson tensor is replaced by a Nambu-Poisson one.\footnote{Let us notice, that in our approach to noncommutativity of fivebrane, the ordinary point-wise product remains undeformed}  It turned out that this requirement determines the bosonic part of the effective action essentially uniquely. 

In our derivation of the action, generalized geometry played an essential role. All key ingredients, have their origin in the generalized geometry. It already has been  appreciated in the literature that the presence of a (p+1)-form leads to a generalized tangent space $TM\oplus \Lambda^p T^*M$. Although, this observation perfectly applies also in our situation, we found it very useful to double it, i.e., to consider the the extended/doubled generalized tangent space $W\oplus W^*$, with $W=TM\oplus \Lambda^p TM$. 

Let us comment on this more: In the string case, $p$=1, the sum of the background fields $g+B$ plays a prominent role. It enters naturally the Polyakov action, the DBI action, Buscher's rules, etc. In generalized geometry, one way define a generalized metric, is to give a subbundle of the generalized tangent bundle $TM\oplus T^*M$ of maximal rank, on which the natural (+) pairing on generalized tangent bundle is positive definite. Such a subbundle can be characterized as a graph of the map from $TM\to T^*M$ defined by the sum $g+B$. Therefore, it is quite natural to look for a formalism which would allow for a natural ``sum" of a metric and a higher rank $(p+1)$-form. What this sum should be is indicated by the Polyakov type membrane action in its matrix form (\ref{def_actiontotmatrix}). From here it is just a small step to recognize the doubled generalized tangent bundle as a right framework for a meaningful interpretation of the ``sum" of the metric and a higher rank $(p+1)$-form. This observation is further supported by the form of the open closed relations in the doubled form (\ref{eq_occorrespondece0}) and the matrix form of the Nambu sigma model (\ref{def_actionnsm}). Finally, the corresponding Hamiltonian (\ref{def_polyakovham}), cf. also (\ref {eq_nsmham}), tells us what the relation to the generalized metric on $TM\oplus\Lambda^p T^\ast M$ is. Hence, at the end, we do not really use the full doubled generalized tangent bundle, we use it only for a nice embedding of the generalized tangent bundle $TM\oplus\Lambda^p T^\ast M$.\footnote{The doubled generalized geometry formalism can also be introduced for the $p$=1 string case and allows an elegant formulation of the theory. For any $p$, the appearance of $T M$ and $\Lambda^p T M$ (and similarly of $T^\ast M$ and $\Lambda^p T^\ast M$) is related to the split into one temporal and $p$~spatial world-sheet directions.}

Nevertheless, we found the doubled generalized geometry quite intriguing. Extending on the above comments: Since on the doubled generalized tangent bundle there is a natural function-valued non-degenerated  pairing $\langle.,.\rangle$, we can mimic the standard constructions with $TM\oplus T^*M$. For instance, one can speak of the orthogonal group, define the generalized metric using an involutive endomorphism $\mathcal{T}$ on $W\oplus W^*$, such that $\langle \mathcal{T},.\rangle$ defines a fibre-wise metric on the doubled generalized tangent bundle, etc. 

However, we are still facing a problem; We lack a canonical Courant algebroid structure. The reason lies basically in very limited choices for the anchor map $\rho: W\oplus W^* \rightarrow TM$, which leave us only with a projection onto the tangent bundle $TM$. The map $\rho$ is therefore ``too simple" to control the symmetric part of any bracket. However, we can still consider Leibniz algebroid structures on $W \oplus W^{\ast}$. 
There are several possibilities to do this. To choose the one suitable for $p$-brane backgrounds, one can consider the action of the map $e^{\B}: W \oplus W^{\ast} \rightarrow W \oplus W^{\ast}$, where $\B$ is a general section of $\Lambda^{2}W$, viewed as a map from $W$ to $W^{\ast}$, and extended to $\mbox{End}(W \oplus W^{\ast})$ by zeros. The map $e^{\B}$ is thus an analogue of the usual $B$-field transform of generalized geometry $TM \oplus T^{\ast}M$. It turns our that there is a Leibniz algebroid, such that the condition for $e^{\B}$ to be an isomorphism of the bracket forces $\B$ to take the block off-diagonal form (\ref{eq_bigGbigB}), with $C \in \Omega^{p+1}_{closed}(M)$. This bracket coincides with the one defined by Hagiwara in \cite{hagiwara} to study Nambu-Dirac manifolds. Moreover, Nambu-Poisson manifolds appear naturally as its Nambu-Dirac structures. Interestingly, its full group of orthogonal automorphisms can be calculated, giving (for $p>1$) a semi-direct product $\Diff(M) \ltimes ( \Omega^{p+1}_{closed} \rtimes G)$, where $G$ is the group of locally constant non-zero functions on $M$. Notably, this coincides with the group of all automorphisms of higher Dorfman bracket, see e.g. \cite{2011ScChA..54..437B}. 

Relating our approach, based on the generalized geometry on the vector bundle $W \oplus W^{\ast}$, with the usual generalized geometries in $M$-theory and supergravity \cite{Hull:2007zu, Berman:2010is, Coimbra:2011nw, Coimbra:2012af}, we notice the following. A choice of a generalized geometry is subject to the field content one wants to describe. In principle, one can double each of of them and use the advantages of having a natural function-valued pairing as we did for our case of interest in this paper.  However, the field content coming with such a doubled generalized geometry is much bigger then we started with and we have to reduce it accordingly. 

Finally, let us again notice the striking similarity with the result of \cite{Cederwall:1997gg}, \cite{Bao:2006ef} -- based on a very different approach -- and  discussed after equation (\ref{Cederwall}). We find worth to pursue a deeper understanding of this similarity in the future.

\section*{Acknowledgement}

It is a pleasure to thank Tsuguhiko Asakawa, Peter Bouwknegt, Chong-Sun Chu, Pei-Ming Ho, Petr Ho\v rava, Dalibor Kar\'asek, Noriaki Ikeda, Matsuo Sato, Libor \v Snobl, and Satoshi Watamura for helpful discussions. B.J. and P.S. appreciate the hospitality of the Center for Theoretical Sciences, Taipei, Taiwan, R.O.C. B.J. thanks CERN for hospitality.
We gratefully acknowledge financial support by the grant GA\v CR P201/12/G028 (B.J.), by the Grant Agency of the Czech Technical University in Prague, grant No. SGS13/217/OHK4/3T/14 (J.V.), and by the DFG within the Research Training Group 1620 ``Models of Gravity'' (J.V., P.S.).
We thank the DAAD (PPP) and ASCR \& MEYS (Mobility) for supporting our collaboration. We also thank the referee for his comments which helped to improve the manuscript.

\appendix
\section{Nambu-Poisson structures} \label{sec_NPstructures}
Here we recall some fundamental properties of Nambu-Poisson
structures \cite{Takhtajan:1993vr} as needed in this paper. For
details see, e.g., \cite{hagiwara} or \cite{2011ScChA..54..437B}.

For any $(p+1)$-vector field $A$ on $M$ we define the induced map
$A^{\sharp}: \Omega^{p}(M) \rightarrow \mathfrak{X}(M)$ as
$A^{\sharp}(\xi) = (-1)^{p} \io_{\xi}A = \xi_{K} A^{iK}
\partial_{i}$.

Also, for an alternative formulation of the fundamental identity, we
need to recall the Dorfman bracket, i.e., the $\R$-bilinear bracket
on the sections of $TM \oplus \cTM{p}$, defined as
\begin{equation} \label{def_dorfman}
[V+\xi,W+\eta]_{D} = [V,W] + \mathcal{L}_{V} \eta - \io_{W} d\xi,
\end{equation}
for all $V,W \in \vf{}$ and $\xi,\eta \in \df{p}$.

Let $\Pi$ be a $(p+1)$-vector field on $M$. We call $\Pi$ a
Nambu-Poisson structure if
\begin{equation} \label{def_npfi} \mathcal{L}_{\Pi^{\sharp}(df_1 \^ \dots
\^ df_p)}(\Pi) = 0\,, \end{equation}
for all $f_1, \dots, f_p \in C^{\infty}(M)$.

\begin{lemma} \label{lem_nplemma}
For an arbitrary $p \geq 1$, the condition (\ref{def_npfi}) can be
stated in the following equivalent ways:
\begin{enumerate}
\item \label{npfieq} The graph $G_{\Pi} = \{ \Pi^{\sharp}(\xi) + \xi \ | \
\xi \in \Omega^{p}(M) \}$ is closed under the Dorfman bracket
    (\ref{def_dorfman});
\item for any $\xi,\eta \in \Omega^{p}(M)$ it holds that
\begin{equation} \label{npfieq1} (\mathcal{L}_{\Pi^{\sharp}(\xi)}
(\Pi))^{\sharp}(\eta) = - \Pi^{\sharp}(\io_{\Pi^{\sharp}(\eta)}(d\xi))\,;
\end{equation}
\item let $[\cdot,\cdot]_{\pi}: \Omega^{p}(M) \times \Omega^{p}(M)
\rightarrow \Omega^{p}(M)$ be defined as
\begin{equation} \label{def_leibnizpi} [\xi,\eta]_{\pi} :=
\mathcal{L}_{\Pi^{\sharp}(\xi)}(\eta) - \io_{\Pi^{\sharp}(\eta)}(d\xi)\,,
\end{equation}
for all $\xi,\eta \in \Omega^{p}(M)$. Then it holds that
\begin{equation} \label{npfieq2} [\Pi^{\sharp}(\xi),\Pi^{\sharp}(\eta)] =
\Pi^{\sharp}([\xi,\eta]_{\pi})\,, \end{equation} for all $\xi, \eta
\in \Omega^{p}(M)$; \item for any $\xi \in \Omega^{p}(M)$ it holds
that
\begin{equation} \label{npfieq3} \mathcal{L}_{\Pi^{\sharp}(\xi)}(\Pi) =
- \big( \io_{d\xi}(\Pi) \Pi - \frac{1}{p+1} \io_{d\xi}(\Pi \^ \Pi)
\big)\,. \end{equation}
\end{enumerate}
\end{lemma}

For $p>1$, around any point $x \in M$, where $\Pi(x) \neq 0$, there
exist local coordinates $(x^{1},\dots,x^{n})$, such that
\begin{equation} \label{eq_Pidecomposition}
\Pi(x) = \ppx{}{1} \^ \cdots \^ \ppx{}{p+1}.
\end{equation}
In this coordinates $\Pi^{iJ} = \delta^{iJ}_{1 \dots p+1} =
\epsilon^{iJ}$.

For $p>1$, a Nambu-Poisson tensor can be multiplied
by any smooth function, and one gets again a Nambu-Poisson tensor:

\begin{lemma} \label{lem_nptensormultiple}
Let $\Pi$ be a Nambu-Poisson tensor, and $p>1$. Let $f \in
C^{\infty}(M)$ be a smooth function on $M$. Then $f \Pi$ is again  a
Nambu-Poisson tensor. For $p=1$ this is not true in general.
\end{lemma}

This lemma has a simple useful consequence
\begin{lemma} \label{lem_nptoplevel}
Let $n = p+1$. Then any $\Pi \in \Gamma(\Lambda^{p+1}TM)$ is a Nambu-Poisson tensor.
\end{lemma}

There is an interesting little  technical detail. One of the equivalent reformulations of fundamental identity was the closedness of the graph $G_{\Pi}$ under the Dorfman bracket. But see that the both, the definition of $G_{\Pi}$ and the  involutivity condition have a good meaning also for \emph{any} vector bundle morphism $\Pi^{\sharp}: \cTM{p} \rightarrow TM$. We may ask whether there exists $\Pi^{\sharp}$, which is not induced by $(p+1)$-vector on $M$. The answer is given by the following lemma:

\begin{lemma} \label{lem_npskewsym}
Let $\Pi^{\sharp}: \cTM{p} \rightarrow TM$ be a vector bundle morphism, such that its graph
\[ G_{\Pi} = \{ \Pi^{\sharp}(\xi) + \xi \ | \ \xi \in \Omega^{p}(M) \}, \]
is closed under higher Dorfman bracket (\ref{def_dorfman}). Let $\Pi$ be a contravariant $(p+1)$-tensor defined by
\[ \Pi(\alpha,\xi) = \<\alpha, \Pi^{\sharp}(\xi) \>, \]
for all $\alpha \in \df{1}$ and $\xi \in \df{p}$. Then $\Pi$ is a $(p+1)$-vector, and hence a Nambu-Poisson tensor.
\end{lemma}
\begin{proof}
The closedness of $G_{\Pi}$ under the Dorfman bracket can immediately be rewritten as (\ref{npfieq1}), where $\Pi$ is now not necessarily a $(p+1)$-vector. This relation is tensorial in $\eta$, so choose $\eta = dy^{J}$, and look at the $i$-th component of the identity. The left-hand side is
\[
\begin{split}
(\mathcal{L}_{\Pi^{\sharp}(\xi)}\Pi)^{iJ} & = \xi_{K} \Big( \Pi^{mK} {\Pi^{iJ}}_{,m} - {\Pi^{iK}}_{,m} \Pi^{mJ} - \sum_{r=1}^{p} {\Pi^{j_{r}K}}_{,m} \Pi^{ij_{1} \dots m \dots j_{p}} \Big) \\
& - \xi_{K,m} \Big( \Pi^{iK} \Pi^{mJ} + \sum_{r=1}^{p} \Pi^{j_{r}K} \Pi^{i j_{1} \dots m \dots j_{p}} \big).
\end{split}
\]	
The right-hand side of (\ref{npfieq1}) is
\[ - \Pi^{\sharp}( \io_{\Pi^{\sharp}(dy^{J})} d\xi)^{i} = \Pi^{iM} \Pi^{lJ} (d\xi)_{lM} = - \xi_{K,m} \big( \Pi^{iM} \Pi^{lJ} \delta_{lM}^{mK} \big). \]
The terms proportional to $\xi_{K}$ form the differential part of the identity, whereas the terms proportional to $\xi_{K,m}$ form the algebraic part:
\[ \Pi^{iK} \Pi^{mJ} + \sum_{r=1}^{p} \Pi^{j_{r}K} \Pi^{ij_{1} \dots m \dots j_{p}} = \Pi^{iM} \Pi^{lJ} \delta_{lM}^{mK}. \]
We will use this algebraic identity to show that $\Pi^{kM} = 0$, whenever $k \in M$. This will prove that $\Pi$ is a $(p+1)$-vector. To do this, choose $m=i=k$, and $K = J = M$ in the above identity. Assume that $m_{q} = k$, where $M = (m_{1} \dots m_{p})$. Then, the only non-trivial term in the sum is the one for $r=q$. Right-hand side vanishes due to skew-symmetry of the symbol $\delta$. Hence, we obtain
\[ 2(\Pi^{kM})^{2} = 0. \]
This proves that $\Pi^{kM} = 0$, and $\Pi$ is thus a $(p+1)$-vector.
\end{proof}

\subsection{Scalar density}
Interestingly, the coordinates $(x^1, \dots x^n)$, in which $\Pi$ has the form (\ref{eq_Pidecomposition}), allow us to define a well-behaved scalar density $|\Pi(x)|$ of weight $-(p+1)$. Let $(y^{1},\dots,y^{n})$ be arbitrary local coordinates. Define the function $|\Pi(x)|$ as
\begin{equation} \label{def_pidensity}
|\Pi(x)| = \det{ \left( \frac{\partial y^{i}}{\partial x^{j}} \right) }^{p+1},
\end{equation}
that is, the Jacobian of the coordinate transformation $y^{i} = y^{i}(x^{k})$. This is indeed a scalar density (with respect to a change $y \mapsto \~y$) of weight $-(p+1)$, as can easily be seen using the chain rule.

For $p=1$, let $\Pi^{ij}$ be the matrix of $\Pi$ in $(y)$ coordinates. We can ask, whether $|\Pi(x)| = \det{\Pi^{ij}}$ whenever $\Pi$ is decomposable. The answer is clearly negative for $n > 2$, where $\det{\Pi^{ij}} = 0$. The case $p=1$, $n=2$ is a special case contained in the next question. Let $p \geq 1$ and $n = p+1$. Let $\Pi^{iJ}$ be the matrix of the vector bundle map $\Pi^{\sharp}$. For $n = p+1$, this is a square $n \times n$ matrix. We can thus ask whether $|\Pi(x)| = \det{ \Pi^{iJ} }$. It is of course modulo the sign, depending on the ordering of the basis of $\Omega^{p}(M)$. Now, see that
\[ \Pi(x) = \ppx{}{1} \^ \cdots \^ \ppx{}{p+1} = |\Pi(x)|^{\frac{1}{n}} \ppy{}{1} \^ \cdots \^ \ppy{}{p+1}. \]
This means that $|\Pi(x)|^{\frac{1}{n}} = \Pi^{1 \dots n}(x)$. The determinant of $\Pi^{iJ}$ is up to sign the $n$-th power of $\Pi^{1 \dots n}$, and thus $\det{ \Pi^{iJ} } = \pm |\Pi(x)|$.

Further, we have to be careful with the dependence of $|\Pi(x)|$ on the choice of the special local coordinates $(x^{1},\dots,x^{n})$. Let $(x'^{1},\dots,x'^{n})$ is another set of such coordinates, that is
\begin{equation} \label{eq_Piintwocoordinates}
 \Pi(x) = \ppx{}{1} \^ \cdots \^ \ppx{}{p+1} = \frac{\partial}{\partial x'^{1}} \^ \cdots \^ \frac{\partial}{\partial x'^{p+1}}.
\end{equation}
Denote by $J$ the Jacobi matrix of the transformation $\~x^{i} = \~x^{i}(x^{k})$. We can split it as
\[ J = \bm{J_{\epsilon}}{K}{L}{M}, \]
where the top-left block $J_{\epsilon}$ is a $(p+1)\times(p+1)$ submatrix corresponding to the first $p+1$ of both sets of coordinates. The condition in (\ref{eq_Piintwocoordinates}) forces $\det{(J_{\epsilon})} = 1$ and $L = 0$. We thus get the important observation that
\[ \det{J} = \det{M}, \]
and moreover $\det{M} = \det{M}(x^{j>p+1})$. This implies that $|\Pi(x)|$ transforms, with respect to the change the special coordinates $(x)$, as
\begin{equation} \label{eq_Pidensitytransform}
|\Pi(x)| = \det{(M)}^{p+1} |\Pi(x)|',
\end{equation}
where $|\Pi(x)|'$ is calculated with respect to $(x')$ coordinates on $M$.
\section{Background independent gauge} \label{sec_BIGsupplement}
\subsection{Pseudoinverse of a $2$-form}
\begin{tvrz} \label{tvrz_2formpseudoinverse}
Let $V$ be a finite-dimensional vector space. Let $g$ be an inner product on $V$, and $C \in \Lambda^{2}V^{\ast}$ a $2$-form on $V$. Let $P: V \rightarrow V$ an orthogonal projector, such that $\ker(P) = \ker(C)$. Then there exists a unique $2$-vector $\Pi$, such that
\[ \Pi C = P \ , \ P \Pi = \Pi. \]
\end{tvrz}
\begin{proof}
Let $\mathbf{C}$, $\mathbf{g}$ and $\mathbf{P}$ be the matrices of $C$, $g$, $P$, respectively, in an arbitrary fixed basis of $V$. First construct the map $\~C \equiv g^{-1}C: V \rightarrow V$. This map is skew-symmetric with respect to $g$. Indeed, we have
\[ \mathbf{g}^{-1} ( \mathbf{g}^{-1} \mathbf{C} )^{T} \mathbf{g} = -\mathbf{g}^{-1} \mathbf{C}. \]
Denote $\mathbf{\~C} = \mathbf{g}^{-1} \mathbf{C}$. Let $\mathbf{A}$ be the matrix diagonalizing $\mathbf{g}$, that is $\mathbf{A}^{T} \mathbf{g} \mathbf{A} = \mathbf{1}$. Finally, define the matrix $\mathbf{\~C'} = \mathbf{A}^{-1} \mathbf{\~C} \mathbf{A}$. This matrix is skew-symmetric (in the ordinary sense). Standard linear algebra says that there exists a standard block-diagonal form of the matrix $\mathbf{\~C'}$. In more detail, one can find an orthogonal matrix $\mathbf{O}$ and a matrix $\mathbf{\Sigma}$, such that $\mathbf{\~C'} = \mathbf{O} \mathbf{\Sigma} \mathbf{O}^{T}$, where $\mathbf{\Sigma}$ has the form
\[ \mathbf{\Sigma} = \diag \Big( \bm{0}{\lambda_{1}}{-\lambda_{1}}{0}, \dots, \bm{0}{\lambda_{k}}{-\lambda_{k}}{0}, 0, \dots, 0 \Big).
\]
where $k = \frac{1}{2} \mbox{rank}(\mathbf{\~C'})$, and $\lambda_{1}, \dots, \lambda_{k} > 0$. Note that the matrix $\mathbf{O}$ is not unique, and the matrix $\mathbf{\Sigma}$ is unique up to the reordering of the $2 \times 2$ blocks.

This shows that we can write $\mathbf{C} = \mathbf{g} \mathbf{A} \mathbf{O} \mathbf{\Sigma} \mathbf{O}^{T} \mathbf{A}^{-1}$. Define $\mathbf{\Delta}_{2k} = \diag(1, \dots, 1, 0, \dots, 0)$, where the number of $1$'s is $2k$. The (unique) matrix $\mathbf{P}$ can be now written as $\mathbf{P} = \mathbf{A} \mathbf{O} \mathbf{\Delta}_{2k} \mathbf{O}^{T} \mathbf{A}^{-1}$.
Let $\mathbf{\Pi}$ be the matrix of a bivector we are looking for. The equation $\Pi C = P$ translates into
\[ \mathbf{\Pi g A O \Sigma} \mathbf{O}^{T} \mathbf{A}^{-1} = \mathbf{A} \mathbf{O} \mathbf{\Delta}_{2k} \mathbf{O}^{T} \mathbf{A}^{-1}. \]
We thus get that $ (\mathbf{O}^{T} \mathbf{A}^{-1} \mathbf{\Pi} \mathbf{g} \mathbf{A} \mathbf{O} ) \mathbf{\Sigma} = \mathbf{\Delta}_{2k}$. This means that
\[ \mathbf{\Pi} = \mathbf{A} \mathbf{O} \mathbf{\Sigma}^{+} \mathbf{O}^{T} \mathbf{A}^{-1} \mathbf{g}^{-1}, \]
where $\mathbf{\Sigma}^{+} \mathbf{\Sigma} = \mathbf{\Delta}_{2k}$. Now it is easy to see that $\Pi$ is a bivector, if and only if $\mathbf{\Sigma}^{+}$ is, and that $P \Pi = \Pi$ holds if and only if $\mathbf{\Delta_{2k}} \mathbf{\Sigma}^{+} = \mathbf{\Sigma}^{+}$. This fixes $\mathbf{\Sigma}^{+}$ and thus $\mathbf{\Pi}$ uniquely. It coincides with the Moore-Penrose pseudoinverse of the matrix $\mathbf{\Sigma}$, and it is given, in the block form, as
\[ \mathbf{\Sigma}^{+} = \bm{\mathbf{\Sigma_{0}}^{-1}}{0}{0}{0}, \]
where $\mathbf{\Sigma_{0}}$ is the invertible top left $2k \times 2k$ block of $\mathbf{\Sigma}$.
\end{proof}

\subsection{Integrable forms}
Let $M$ be a smooth manifold, and let $C$ be a $(p+1)$-form on $M$. The form $C$ is called an integrable form if it holds
\begin{equation} \label{def_iforms1} C(\mathbf{P}) \^ C = 0, \end{equation}
\begin{equation} \label{def_iforms2} C(\mathbf{P}) \^ dC = 0, \end{equation}
for all $\mathbf{P} \in \vf{p}$, where on the left-hand side $C(\mathbf{P})$ denotes the value of the induced  vector bundle morphism $C: \Lambda^{p} TM \rightarrow T^{\ast}M$ when evaluated on $(\mathbf{P})$. The condition (\ref{def_iforms1}) is in fact a very restrictive one. Also, it is very similar to the algebraic part of Nambu-Poisson fundamental identity:
\begin{lemma}
Let $C$ be a $(p+1)$-form. Then $C$ satisfies (\ref{def_iforms1}) if and only if it is decomposable around every point $x \in M$, such that $C(x) \neq 0$. That means that there exists a $(p+1)$-tuple $(\alpha_{1}, \dots, \alpha_{p+1})$ of linearly independent $1$-forms , such that locally
\[ C = \alpha_{1} \^ \dots \^ \alpha_{p+1}. \]
\end{lemma}
\begin{proof}
Let us proceed by induction on $p$.
The $p=0$ case is a trivial statement, any $1$ form is decomposable. Now choose $p>0$.
Assume that statement holds for all $p$-forms, and let $C$ be a $(p+1)$-form satisfying (\ref{def_iforms1}). We have to show that it is decomposable.

Let $x \in M$, such that $C(x) \neq 0$.
First, see that for any $V \in \vf{}$, such that $(\io_{V}(C))(x) \neq 0$, the $p$-form $\io_{V}(C)$ satisfies (\ref{def_iforms1}), and thus, by induction hypothesis, is decomposable. Let us take any $\mathbf{Q} \in \vf{p-1}$. We have to show that
\[ (\io_{V}C)(\mathbf{Q}) \^ (\io_{V}C) = 0. \]
But this can be rewritten as
\[ \io_{V}\big( C(V \^ \mathbf{Q}) \^ C \big) = 0, \]
which follows from the assumptions on $C$, taking $\mathbf{P} = V \^ \mathbf{Q}$. Second, take the original condition (\ref{def_iforms1}) and apply $\io_{V}$ to both sides with an arbitrary $V \in \vf{}$. One gets
\[ \io_{V}(C(\mathbf{P})) \cdot C - C(\mathbf{P}) \^ \io_{V}(C) = 0. \]
But $\io_{V}(C(\mathbf{P}))$ is a scalar function, and since $C$ is a nonzero $(p+1)$-form at $x$, there have to exist $V \in \vf{}$ and $\mathbf{P} \in \vf{p}$, such that $\lambda \equiv \io_{V}(C(\mathbf{P})) \neq 0$, at least at some neighborhood of $x$.  Thus, locally we can write
\[ C = \frac{1}{\lambda} C(\mathbf{P}) \^ \io_{V}(C). \]
Since $\lambda(x) \neq 0$, also $(\io_{V}(C))(x) \neq 0$.  We can now apply the induction hypothesis to this $p$-form to get $p$ linearly independent $1$-forms $(\alpha_{1}, \dots, \alpha_{p})$, such that
\[ \io_{V}C = \alpha_{1} \^ \dots \^ \alpha_{p}. \]
This finishes the proof, because taking $\alpha_{p+1} = \frac{(-1)^{p}}{\lambda} C(\mathbf{P})$ leads to the desired decomposition.
\end{proof}

Let us now clarify where integrable forms got their name from:

\begin{definice}
Let $C$ is a $(p+1)$-form. Denote by $M'$ the open submanifold of $M$, where $C \neq 0$. The kernel distribution $K$ of $C$ is a distribution on $M'$, defined at every $x \in M'$ as
\[ K_{x} = \{ V \in T_{x}M \ | \ \io_{V}(C(x)) = 0 \}. \]
Note that this distribution is not necessarily a smooth one.
\end{definice}

We can now relate integrability of distributions to the integrability od forms.

\begin{lemma}
Let $C$ be a $(p+1)$-form.
Then $C$  integrable if and only if $K$ is an integrable $(n-(p+1))$-dimensional regular smooth distribution on $M'$.
\end{lemma}

\begin{proof}
First assume that $C$ is an integrable $(p+1)$-form. Then by the previous lemma, around every point of $x \in M'$, there exists a $(p+1)$-tuple of linearly independent $1$-forms, such that locally
\begin{equation} \label{eq_Cdecomposed}
 C = \alpha_{1} \^ \dots \^ \alpha_{p+1}.
\end{equation}
The subspace $K_{x}$ can be determined easily as
\[ K_{x} = \{ V \in T_{x}M \ | \ \io_{V}(\alpha_{i}(x)) = 0, \forall i \in \{1, \dots, p+1\} \}. \]
This is a set of $k$ linearly independent linear equations for the components of $V$. The dimension of $K_{x}$ is thus $n - (p+1)$. To see that this is a smooth regular distribution, note that $K$ is the kernel of a smooth vector bundle morphism of a constant rank, and hence a subbundle of $TM'$. Hence, a smooth distribution in $M'$.

To see that it is also integrable, plug the expression (\ref{eq_Cdecomposed}) into the second defining equation (\ref{def_iforms2}). It turns out that it is equivalent to
\begin{equation} \label{eq_dalphaterms}
d\alpha_{j} \^ \alpha_{1} \^ \dots \^ \alpha_{p+1} = 0,
\end{equation}
for all $j \in \{1, \dots, p+1\}$. Now take any $V \in \Gamma(K)$, and plug it into (\ref{eq_dalphaterms}). It gives $\io_{V}(d\alpha_{j}) = 0$ for all $j \in \{1, \dots, p+1\}$. But this is, using the Cartan formula for $d\alpha_{j}$, equivalent to involutivity of the subbundle $K$ under the commutator of vector fields, which is in turn, using the Frobenius integrability theorem, equivalent to the integrability of $K$.

Conversely, assume that $K$ is integrable $((n-(p+1))$-dimensional regular smooth distribution. At every $x \in M'$, there is a neighborhood $U_{x} \ni x$, and a set of local coordinates $(x^{1}, \dots, x^{(n-(p+1))}, y^{1}, \dots, y^{p+1})$, such that sections of the subbundle $K$ are on $U_{x}$ spanned by $(\ppx{}{1}, \dots, \ppx{}{(n-(p+1))})$. Then $C$ has to be annihilated by all vectors of $K$, so it has to have the local form
\begin{equation} \label{eq_Clocaliny}
 C = \lambda \cdot dy^{1} \^ \dots \^ dy^{p+1}.
\end{equation}
We see that this $C$ clearly satisfies (\ref{def_iforms1}). Since we are on $M'$, we have $\lambda \neq 0$. We set $\alpha_{1} = \lambda dy^{1}$, and  $\alpha_i=dy^{i}$ for $i=2,\dots, p+1$. The second condition for integrable $(p+1)$-forms translates as (\ref{eq_dalphaterms}). Obviously, this holds for the above defined $\alpha_{j}$'s.

At $x \in M \setminus M'$ the integrability conditions (\ref{def_iforms1}, \ref{def_iforms2}) hold trivially and we can conclude that $C$ is an integrable $(p+1)$-form.
\end{proof}

\begin{rem}
One can extend the distribution $K$ to the whole manifold $M$. For each $x \in M \setminus M'$, define $K_{x} = \{0\}$. By this extension one gets a smooth singular distribution on $M$. However, even for integrable $(p+1)$-forms, $K$ is not integrable in general. For details see \cite{dufour2005poisson}.
\end{rem}

Let us conclude this section by relating the concepts of integrable $(p+1)$-forms to Nambu-Poisson structures. This is given by the following lemma.

\begin{lemma}
Let $M$ be an orientable smooth manifold. Let $\Omega$ be the corresponding volume form.
Let $C$ be a $(p+1)$-form on $M$. Define a $(p+1)$-vector $\Pi$ by equation
\[ \io_{\Pi} \Omega = C. \]
Then $\Pi$ is a Nambu-Poisson $(n-(p+1))$-vector if and only if $C$ is an integrable $(p+1)$-form.
\end{lemma}

\begin{proof}
Clearly, $\Pi(x) = 0$ if and only if $C(x) = 0$. Let $\Pi$ be a Nambu-Poisson tensor. By previous comment, at singular points of $\Pi$, $C$ vanishes. The conditions on integrability are, at these points, satisfied trivially. Assume that $\Pi(x) \neq 0$. Then there exist local coordinates $(x^{1}, \dots, x^{n})$ around $x$, such that
\[ \Pi = \ppx{}{1} \^ \dots \^ \ppx{}{n-(p+1)}. \]
In these coordinates, the volume form $\Omega$ is
\[ \Omega = \omega \cdot dx^{1} \^ \dots \^ dx^{n}, \]
where $\omega \neq 0$. We thus see that $C$ has the explicit form
\[ C = \omega \cdot dx^{n-(p+1)+1} \^ \dots \^ dx^{n}. \]
It is easy to check that it satisfies both integrability conditions (\ref{def_iforms1},\ref{def_iforms2}).

The converse statement follows basically from the proof of the previous lemma. There, we have shown that $C$ can be, for an integrable $(p+1)$-form, written (around any point where $C(x) \neq 0$) in the local form (\ref{eq_Clocaliny}). Writing the volume form in these local coordinates as
\[ \Omega = g \cdot dx^{1} \^ \dots dx^{(n-(p+1))} \^ dy^{1} \^ \dots dy^{p+1}, \]
one finds the local expression for $\Pi$ as
\[ \Pi = \frac{\lambda}{g} \cdot \ppx{}{1} \^ \dots \^ \ppx{}{n-(p+1)}. \]
Note that this is a top-level multivector field on the submanifold $N'$. In the view of lemma \ref{lem_nptoplevel}, one would expect that this is enough. Inspection of the fundamental identity shows that all partial derivatives are contracted with the components of $\Pi$, so in the fundamental identity there are no partial derivatives in transversal directions. We can now apply (the proof of) lemma \ref{lem_nptoplevel} to conclude that $\Pi$ is a Nambu-Poisson tensor on $M$.
\end{proof}
\bibliography{nambu_sigma}
\end{document}